\documentclass[aps,reprint,superscriptaddress,nofootinbib,floatfix]{revtex4-2}
\usepackage{amsmath,amssymb,amsfonts,mathtools,bm}
\usepackage{amsthm}
\usepackage{graphicx}
\usepackage{booktabs,array}
\usepackage{orcidlink}
\usepackage{xcolor}
\usepackage{hyperref}
\hypersetup{colorlinks=true,linkcolor=blue!55!black,citecolor=blue!55!black,urlcolor=blue!55!black}
\graphicspath{{./}}

\newtheorem{theorem}{Theorem}[section]
\newtheorem{proposition}[theorem]{Proposition}
\newtheorem{lemma}[theorem]{Lemma}
\newtheorem{corollary}[theorem]{Corollary}
\newtheorem{definition}[theorem]{Definition}
\newtheorem{remark}[theorem]{Remark}

\newcommand{\Del}{\mathrm{Del}}
\newcommand{\sym}{\mathrm{sym}}
\newcommand{\Tr}{\operatorname{Tr}}
\newcommand{\Bin}{\mathrm{Bin}}
\newcommand{\E}{\mathbb{E}}
\newcommand{\id}{\mathrm{id}}
\newcommand{\ket}[1]{\lvert #1\rangle}
\newcommand{\bra}[1]{\langle #1\rvert}

\newcommand{\cE}{\mathcal E}
\newcommand{\cH}{\mathcal H}

\newcommand{\cR}{\mathcal R}

\newcommand{\cL}{\mathcal L}
\newcommand{\etaG}{\eta_{\rm G}}

\newcommand{\cC}{\mathcal C}
\newcommand{\cD}{\mathcal D}
\newcommand{\cP}{\mathcal P}
\newcommand{\cT}{\mathcal T}
\newcommand{\norm}[1]{\left\lVert #1\right\rVert}
\newcommand{\Rdel}{R}
\newcommand{\Jsurv}{M}
\newcommand{\Mbin}{\mathsf M}
\newcommand{\Nfp}{N}
\newcommand{\epsq}{\varepsilon_q}

\newcommand{\ContractX}{\mathsf X}

\begin{document}

\title{Recovery-Free CHSH Nonlocality with Particle Loss}

\author{Giovanni Scala~\orcidlink{0000-0003-2685-0946}}
\affiliation{Dipartimento Interateneo di Fisica, Politecnico di Bari, 70126 Bari, Italy}
\affiliation{INFN, Sezione di Bari, 70126 Bari, Italy}
\author{Cosmo Lupo~\orcidlink{0000-0002-5227-4009}}
\affiliation{Dipartimento Interateneo di Fisica, Politecnico di Bari, 70126 Bari, Italy}
\affiliation{INFN, Sezione di Bari, 70126 Bari, Italy}

\begin{abstract}
Can CHSH nonlocality survive particle loss without applying an explicit recovery operation?  
In principle, any deterministic recovery can be absorbed into the measurement. Operationally, we show that the answer depends on the allowed measurements on the lossy system. We distinguish flagged erasure, in which each lost particle leaves a detectable record, from unflagged deletion, in which no such record remains. 
For flagged erasure and survival probability \(\eta>1/2\), using known quantum-capacity results, we show measurements that
asymptotically approach the quantum maximum \(2\sqrt2\). In contrast, for
\(\eta\le1/2\), CHSH violation is impossible for both loss models.
Then, we construct explicit recovery-free protocols using
permutation-invariant encodings built from $n$-qubit Dicke states $| D_N^n\rangle$ and measurements on the surviving particles. 
A one-excitation $(N=1)$ encoding violates CHSH for
\(\eta>1/\sqrt2\). 
Increasing the excitation number $N$ yields a family of protocols that violates CHSH for \(\eta>\eta_G=(\sqrt5-1)/2\), with the asymptotic golden ratio approached as $N \to \infty$. 
Finally, we present a sparse-deletion binomial PI protocol that guarantees CHSH violation for up to $O(\sqrt n)$ deletion errors.
Our results distinguish fundamental limits imposed by loss from
those set by explicit measurements without recovery.
\end{abstract}

\maketitle
\section{Introduction}
\noindent Bell inequalities certify correlations that cannot be explained by local hidden-variable theories. The canonical CHSH test combines the correlations \(E_{xy}\) obtained from two binary measurement choices \(x,y\in\{0,1\}\) for each party into $S=E_{00}+E_{01}+E_{10}-E_{11}$.
Local correlations satisfy \(|S|\le2\), whereas quantum theory allows \(|S|>2\) up to \(|S|\le2\sqrt2\)~\cite{Clauser1969,Cirelson1980,Brunner2014}. Such violations underpin device-independent quantum information~\cite{Acin2006,Brunner2014, primaatmaja2023security,zapatero2023advances,Ghoreishi2025}.
Loss is a major obstacle: it degrades entanglement, changes the number of systems available for measurement, and if no-detection events are discarded, opens the detection loophole~\cite{Pearle1970,Eberhard1993,Larsson2014}. 
Quantum error correction suggests encoding the entangled state redundantly and applying an active recovery before the Bell measurements~\cite{Knill1997,Desef2022,Cai2024}. 
Related approaches avoid active recovery by using quantum error-correcting codes to design Bell tests that tolerate errors~\cite{Walker2008,Chen2009, Silberstein2023,Zhang2026}.
This raises the question: \emph{how much CHSH nonlocality remains accessible by explicitly measuring the damaged state directly, without first reconstructing the encoded state?}
Formally, any deterministic recovery followed by a logical measurement can be written as a single POVM on
the post-loss state. Thus, if arbitrary collective POVMs are allowed, simply calling a Bell test ``recovery-free'' does not restrict the obtainable correlations. The operational question is which explicit measurements can be performed directly on the surviving particles.
The answer also depends on the loss model. In flagged erasure, a lost particle
is replaced by the orthogonal state \(\ket{\varnothing}\), so that its position is
recorded. In unflagged deletion, this record is discarded, and only the surviving
particles remain in a shorter register. These models have different reference
limits. For flagged erasure, the quantum capacity is positive for
\(\eta>1/2\)~\cite{Bennett1997,Lloyd1997,Barnum2000,Devetak2005,Devetak2005a}. Hence, with arbitrary direct POVMs, the CHSH
value can asymptotically approach \(2\sqrt2\). This argument uses the erasure
record and does not by itself apply to unflagged deletion. Conversely, for
\(\eta\le1/2\), we show that both loss models are two-extendible and therefore
cannot violate CHSH~\cite{Terhal2003,Toner2006,Seevinck2010,Cui2025}.
The nontrivial task is therefore to find explicit measurements,
particularly for unflagged deletion, that preserve Bell nonlocality without
implementing a recovery operation.
We address this problem using permutation-invariant (PI) encodings built from Dicke states, the equal superpositions of all \(n\)-qubit configurations with a fixed number of excitations~\cite{Dicke1954,Ouyang2014, NakayamaHagiwara2020,PollatsekRuskai2004}. 
Their permutation symmetry makes deletion independent of which particles are lost~\cite{Ouyang2021,Shibayama2021,Aydin2024}. 
We deliver explicit recovery-free CHSH protocols in two complementary regimes:
low-excitation encodings preserve nonlocality at a constant loss rate, whereas binomial PI encodings use redundancy to tolerate a number of deletions growing as \(\sqrt n\). Together with the two-extendibility bound, these results delineate the operational regimes of recovery-free Bell nonlocality.
These results paves the way to recovery-free Bell tests: they
separate the fundamental no-go region from regimes in which explicit direct
measurements can preserve nonlocality, while clarifying the trade-off between
loss tolerance and measurement complexity.
\section{Operational models}
Let \(\rho\) be the joint state after local loss. Suppose that Alice
and Bob apply deterministic local recovery channels
\(\cR_A,\cR_B\) before measuring POVMs
\(\{M^A_{a|x}\}\) and \(\{M^B_{b|y}\}\). Define
\begin{equation}
\widetilde M^A_{a|x}=\cR_A^\dagger(M^A_{a|x}),
\qquad
\widetilde M^B_{b|y}=\cR_B^\dagger(M^B_{b|y}).
\label{eq:heisenberg}
\end{equation}
Because the recovery channels are trace preserving, their adjoints are
unital, so $\widetilde M^A_{a|x},
\widetilde M^B_{b|y}$ form valid POVMs. Trivially,
\begin{equation*}
    \Tr\!\left[
(M^A_{a|x}\otimes M^B_{b|y})
(\cR_A\otimes\cR_B)(\rho)
\right]=
\Tr\!\left[
(\widetilde M^A_{a|x}\otimes\widetilde M^B_{b|y})\rho
\right].
\end{equation*}
Thus every deterministic recovered strategy, without postselection,
has a direct-POVM representation with identical Bell statistics; see Suppl. Mat., Sec.~S2.
This equivalence concerns statistics, not implementation:
\(\widetilde M^A_{a|x}\) and \(\widetilde M^B_{b|y}\) may be as
difficult to realize as the recovery operations they replace. We
nevertheless allow arbitrary direct POVMs in the next argument.

\section{Universal no-go for \(\eta\le1/2\)}
To rule out CHSH violation by any POVMs, we show that flagged, and then, unflagged
deletion produces a two-extendible state on Bob's side.
Consider an arbitrary state \(\rho_{AB_0}\), where \(A\) is Alice's
system and \(B_0\) is one of Bob's particles before loss. Let
$\sigma_{AB}=\bigl(\id_A\otimes\mathcal E_\eta^{\rm flag}\bigr)
(\rho_{AB_0})$
be the state after flagged erasure of \(B_0\). The state
\(\sigma_{AB}\) is two-extendible on Bob's side if there exists a
tripartite state \(\omega_{ABB'}\), invariant under exchange of the
auxiliary output registers \(B\) and \(B'\), such that
$\omega_{AB}=\omega_{AB'}=\sigma_{AB}$. Thus Alice has the same reduced state with either \(B\) or \(B'\).
Such an extension exists whenever \(\eta\le1/2\)~\cite{Bennett1997,Terhal2003,Toner2006,Cubitt2008}.
Applying this extension independently to all of Bob's particles gives a
symmetric extension of his entire post-erasure register. Unflagged
deletion is obtained by locally removing the erasure symbols and packing
the survivors into a shorter register. Performing this same local
postprocessing on \(B\) and \(B'\) preserves both the exchange symmetry
and the equality of the two Alice--Bob marginals. Hence the state after
unflagged deletion is also two-extendible. Extending Bob's side alone is
sufficient, since any subsequent loss on Alice's side is a local
operation and cannot create a Bell violation.

Two-extendible states admit a local model for Bell tests with at most
two measurement settings on the extended party~\cite{Terhal2003}.
Since CHSH has exactly two settings for Bob, every choice of local
POVMs must therefore satisfy \(S\le2\). Equivalently, identical
marginals imply \(S_{AB}=S_{AB'}\), while CHSH monogamy requires
\(S_{AB}^{\,2}+S_{AB'}^{\,2}\le8\)~\cite{Toner2006}. We conclude that, for
\(\eta\le1/2\), CHSH violation is impossible after either flagged or unflagged deletion, even with arbitrary direct POVMs.

Equation~\eqref{eq:heisenberg} also shows that deterministic recovery before measurement cannot circumvent this no-go result.
However, for flagged erasure, \(\eta=1/2\) is the exact asymptotic boundary with arbitrary POVMs. For every \(\eta>1/2\), the quantum erasure channel has positive quantum capacity~\cite{Bennett1997,Schumacher1996,Lloyd1997,Devetak2005}. Encoding and deterministically decoding one half of an entangled logical state can therefore recover it with fidelity approaching one, yielding \(S\to2\sqrt2\). In this case, by Eq.~\eqref{eq:heisenberg}, equivalent POVMs acting directly on the unrecovered flagged output attain the same correlations; see Supplemental Material, Sec.~S2.

This achievability result does not automatically extend to unflagged deletion. The central constructive problem is consequently to find explicit recovery-free protocols whose threshold lies as close as possible to the universal boundary \(\eta=1/2\) from above.

\section{One-excitation PI protocol}
Recall that \(\ket{D_k^n}\) denotes the symmetric \(n\)-particle Dicke state
with \(k\) excitations. We encode
\[
\ket{\Phi_L^+}
=
\frac{\ket{0_L0_L}+\ket{1_L1_L}}{\sqrt2},
\qquad
\ket{0_L}=\ket{D_0^n},
\quad
\ket{1_L}=\ket{D_1^n}.
\]
Consider a loss branch in which \(r\) particles are deleted from
Alice's and Bob's registers. Both parties then retain $m=n-r$
particles. On the logical subspace, deletion acts as the amplitude-damping
channel \(\mathcal A_{m/n}\): the vacuum is unchanged, while a single
excitation survives with amplitude \(\sqrt{m/n}\). The conditional output is
therefore
\[
\rho_{\rm out}
=
\bigoplus_{m=0}^{n}p_m\,\rho_m,\quad\rho_m
=
\bigl(\mathcal A_{m/n}\otimes\mathcal A_{m/n}\bigr)
\bigl(\ket{\Phi_L^+}\!\bra{\Phi_L^+}\bigr),
\]
with independent particle survival with probability \(\eta\), the common
survivor number is distributed as
$p_m=\binom{n}{m}\eta^m(1-\eta)^{n-m}$.
For \(m\ge1\), define on the surviving logical subspace
\[
Q_m=
\ket{D_0^m}\bra{D_0^m}
+
\ket{D_1^m}\bra{D_1^m},
\]
and
\begin{align}
X_m&=
\ket{D_0^m}\bra{D_1^m}
+
\ket{D_1^m}\bra{D_0^m}, \nonumber\\
Y_m&=
-i\ket{D_0^m}\bra{D_1^m}
+i\ket{D_1^m}\bra{D_0^m}.
\end{align}
Outside \(Q_m\), these observables are set to zero, corresponding to an
unbiased random outcome. Alice and Bob measure the block-diagonal observables
\[
A_0=\bigoplus_{m=0}^{n}X_m,
\quad
A_1=\bigoplus_{m=0}^{n}Y_m,
\qquad
B_{0,1}=\bigoplus_{m=0}^{n}\frac{X_m\pm Y_m}{\sqrt2},
\]
with \(X_0=Y_0=0\). 
These measurements act directly on the post-deletion registers and require no
recovery operation.
In the branch with \(m\) surviving particles on each side,
$S_m=2\sqrt2\,\frac{m}{n}$.
Let \(M\sim\mathrm{Bin}(n,\eta)\) denote the number of particles retained by each party. Averaging over \(M\) gives
\begin{equation}
S_n(\eta)
=
2\sqrt2\,\mathbb E\!\left[\frac{M}{n}\right]
=
2\sqrt2\,\eta\, \Longrightarrow\,\eta>\frac{1}{\sqrt2}.
\label{eq:oneex-threshold}
\end{equation}
Hence this correlated-loss model violates CHSH whenever
\(\eta>1/\sqrt2\); see Supp. Mat S3.

\section{\(N\)-excitation PI protocols}
We now ask whether increasing the excitation number can improve the
one-excitation threshold. We keep
\(\ket{0_L^{(n)}}=\ket{D_0^n}\), but encode the logical one using
\(N\) symmetrically distributed excitations,
\begin{equation}
\ket{0_L^{(n)}}=\ket{D_0^n},
\qquad
\ket{1_L^{(n)}}=\ket{D_N^n}.
\label{eq:finite-n-fockpair}
\end{equation}
To analyze this family, we proceed in three steps: we first describe all
survivor-number sectors in a common excitation-number basis, then show that
finite-\(n\) PI deletion approaches pure loss for fixed \(N\), and finally
construct direct measurements that combine the surviving diagonal correlation
with the remaining coherence.

After deletion, a sector containing \(m\) surviving particles may contain \(j=0,\ldots,\min\{N,m\}\) surviving excitations, represented by \(\ket{D_j^m}\). Although different values of \(m\) correspond to registers of different sizes, we can describe all these states by their common excitation number:
\begin{equation}
J_m\ket{D_j^m}=\ket j,
\qquad
\mathcal H_N=\operatorname{span}\{\ket0,\ldots,\ket N\}.
\label{eq:block-identification}
\end{equation}
This is only a common mathematical description, not a recovery operation: measurements on \(\mathcal H_N\) represent block-diagonal measurements acting on the original survivor sectors.
For fixed \(N\), the finite-\(n\) PI-deletion channel in this representation approaches, as \(n\to\infty\), the following pure-loss channel~\cite{Michael2016,Albert2018}
\begin{equation*}
\mathcal L_\eta(\rho)
=
\sum_{a=0}^{N}L_a\rho L_a^\dagger,
\quad
L_a\ket{k}
=
\sqrt{\binom{k}{a}(1-\eta)^a\eta^{k-a}}\ket{k-a},
\end{equation*}
where \(L_a\ket{k}=0\) for \(a>k\). 
Thus, conditioned on losing \(a\) of the \(k\) excitations, \(k-a\) remain.
For independent particle loss, the number of surviving excitations in a fixed
\(N\)-excitation component is distributed as \(\operatorname{Bin}(N,\eta)\).
When \(N\) is fixed and \(n\) is large, the total number \(m\) of surviving
particles concentrates around \(n\eta\). Hence, the packed PI deletion channel,
restricted to the finite sector spanned by \(\{\ket{D_k^n}:0\le k\le N\}\),
converges in diamond norm to the finite-Fock pure-loss channel \(L_\eta\)
defined above.
This diamond-norm convergence is stable under tensoring with an arbitrary
reference system. It therefore applies directly to the shared encoded Bell
state, where Alice's side may be viewed as the reference for Bob's local loss
channel, and similarly for the two-sided channel. Consequently, the CHSH value
obtained from the corresponding direct finite-\(n\) measurements converges to
the CHSH value of the limiting pure-loss protocol. Any limiting violation with
a nonzero gap above \(2\) therefore persists for all sufficiently large \(n\);
see Supplemental Material, Sec.~S4.

The common representation sends the logical codewords to
\(\ket{0_L}=\ket0\) and \(\ket{1_L}=\ket N\), and hence sends the encoded Bell
state to
\begin{equation}
\ket{\Phi_N^+}
=
\frac{\ket{00}+\ket{NN}}{\sqrt2}.
\label{eq:fockpair-code}
\end{equation}
We call this encoding \emph{vacuum-anchored} because the logical zero is invariant
under loss, whereas \(\ket N\) can decay to
\(\ket{N-1},\ldots,\ket0\).
Partial loss now populates the intermediate levels
\(\ket1,\ldots,\ket{N-1}\). We therefore use on $\cH_N$ one measurement to distinguish
vacuum from nonvacuum and another to probe the remaining
\(\ket0\leftrightarrow\ket N\) coherence:
\begin{equation}
Z_N=\ket0\!\bra0-\sum_{j=1}^{N}\ket j\!\bra j,
\qquad
X_N=\ket0\!\bra N+\ket N\!\bra0.
\label{eq:ZXN}
\end{equation}
On the kernel of \(X_N\), the associated binary POVM produces unbiased random
outputs.

Let $q_N=(1-\eta)^N$ be the probability that all \(N\) excitations are lost on one side. The
\(Z_N\) outcomes disagree only when complete loss occurs on exactly one side
of the \(\ket{NN}\) component. Therefore,
\begin{equation}
\langle Z_N\otimes Z_N\rangle
=
1-2q_N(1-q_N)
\equiv c_N(\eta).
\label{eq:cN}
\end{equation}
Meanwhile, the local coherence
\(\ket0\!\bra N\) is reduced by \(\eta^{N/2}\), giving
\begin{equation}
\langle X_N\otimes X_N\rangle=\eta^N,
\qquad
\langle Z_N\otimes X_N\rangle
=
\langle X_N\otimes Z_N\rangle=0.
\end{equation}

Taking \(A_0=Z_N\), \(A_1=X_N\), and the signed POVM observables
\begin{equation}
B_{0,1}
=
\frac{c_N(\eta)Z_N\pm\eta^N X_N}
{\sqrt{c_N(\eta)^2+\eta^{2N}}},
\end{equation}
gives the directly achievable value
\begin{equation}
S_N(\eta)
=
2\sqrt{c_N(\eta)^2+\eta^{2N}}.
\label{eq:SN}
\end{equation}
The condition \(S_N(\eta)>2\) is equivalent to
\begin{equation}
\eta^{2N}>
4q_N(1-q_N)(1-q_N+q_N^2).
\label{eq:golden-condition}
\end{equation}
For \(1/2\le\eta<1\), dividing by \(q_N\) gives
\begin{equation}
\left(\frac{\eta^2}{1-\eta}\right)^N
>
4(1-q_N)(1-q_N+q_N^2).
\label{eq:golden-comparison}
\end{equation}
This form makes the threshold transparent. If
\(1/2\le\eta\le\etaG\), where
$\etaG=\frac{\sqrt5-1}{2}$,
then \(\eta^2/(1-\eta)\le1\), so the LHS of
Eq.~\eqref{eq:golden-comparison} is at most \(1\). At the same time,
\(q_N\le1/2\), which implies
\[
4(1-q_N)(1-q_N+q_N^2)>1.
\]
The inequality therefore cannot be satisfied for any finite \(N\).
For \(\eta\le1/2\), this conclusion also follows from the universal
no-go result established above.
Conversely, if \(\eta>\etaG\), then
\(\eta^2/(1-\eta)>1\), and the left-hand side of
Eq.~\eqref{eq:golden-comparison} grows exponentially with \(N\).
Since its right-hand side is always smaller than \(4\), choosing a finite
\(N\) such that
\[
\left(\frac{\eta^2}{1-\eta}\right)^N>4
\]
guarantees \(S_N(\eta)>2\). Thus, some finite-\(N\) protocol violates
CHSH exactly when \(\eta>\etaG\).
The boundary is therefore the self-similarity condition
\begin{equation}
\frac{\eta}{(1-\eta)+\eta}=\frac{1-\eta}{\eta},
\end{equation}
whose positive root is
\(\etaG=(\sqrt5-1)/2=\varphi^{-1}\), with the \emph{golden ratio}
\(\varphi=(1+\sqrt5)/2\). Increasing \(N\) suppresses one-sided complete
erasure as \((1-\eta)^N\), but also the squared two-party coherence as
\(\eta^{2N}\). The inverse golden ratio marks where these exponential rates
exchange dominance: above \(\etaG\), complete erasure decays faster, allowing
a finite-\(N\) CHSH violation.
This inverse-golden-ratio threshold is exact for the family of states
superposing the vacuum with an \(N\)-excitation state in the pure-loss limit.
Moreover, any violation predicted in that limit is also achieved by the
corresponding finite-\(n\) PI protocol once \(n\) is sufficiently large; further interpretaion are in
Supplemental Material, Sec.~S4.1.
Figure~\ref{fig:threshold-hierarchy} compares this threshold with
\(\eta=1/\sqrt2\) and the universal boundary \(\eta=1/2\).

\begin{figure}[t]
\centering
\includegraphics[width=0.5\textwidth]{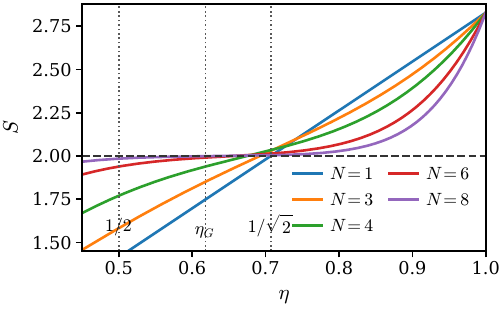}
\caption{Recovery-free CHSH violation under particle loss for the PI protocols. The dashed horizontal line is the local bound \(S=2\).
For the one-excitation encoding (\(N=1\)),
the threshold is
\(\eta=1/\sqrt{2}\).
For the large-\(n\) limit, increasing \(N\) lowers the onset of CHSH violation toward
\(\eta_G=(\sqrt5-1)/2\), the threshold obtained by
optimizing over finite \(N\). We find no protocol such that $S>2$ for $1/2\le \eta \le \eta_G$ (Supp. Mat. S4.4).}
\label{fig:threshold-hierarchy}
\end{figure}
\section{Sparse binomial PI redundancy protocol}
The preceding protocols protect nonlocality at a fixed survival probability by
encoding in a low-excitation sector. We now consider a complementary
finite-size regime. Binomial permutation-invariant (PI) codes use increasing
particle-number redundancy to tolerate a growing \emph{number} of deletions
\cite{Ouyang2014,Aydin2024,Shibayama2021}.

We use the GNU permutation-invariant binomial code~\cite{Ouyang2014,Ouyang2024}, specialized to symmetric Dicke states. Let \(t\) denote the number of deletions to be tolerated, set \(g=t+1\), and choose a code order \(\mathsf M>t\). In the GNU notation, \(g\) is the spacing between occupied Dicke weights, while \(\mathsf M\) determines how many weights
\(0,g,\dots,\mathsf M g\) enter the code and therefore fixes the block size
\(n=g\mathsf M\). The logical states are
\begin{align}
\ket{0_L}&=
\frac{1}{\sqrt{2^{\mathsf M-1}}}
\sum_{\ell\ {\rm even}}\sqrt{\binom{\mathsf M}{\ell}}\ket{D_{g\ell}^{n}},
\nonumber\\
\ket{1_L}&=
\frac{1}{\sqrt{2^{\mathsf M-1}}}
\sum_{\ell\ {\rm odd}}\sqrt{\binom{\mathsf M}{\ell}}\ket{D_{g\ell}^{n}}.
\label{eq:binomial_code_main}
\end{align}
Thus, the logical value is stored in the parity of the coarse index \(\ell\),
while adjacent Dicke components in the two logical superpositions differ by
\(g\) excitations.

Suppose that \(r\le t\) particles are deleted and that \(m=n-r\) particles
remain. In a branch in which \(a\le r<g\) excitations were deleted, a component
\(\ket{D_{g\ell}^{n}}\) contributes to the surviving Dicke weight
\(q=g\ell-a\). Since \(a<g\), the residue class of \(q\) modulo \(g\) uniquely
determines the offset, namely \(a=(-q)\bmod g\), without any physical
correction. More precisely, for
\[
a(q)=(-q)\bmod g,\qquad L(q)=\frac{q+a(q)}{g}.
\]
For every Dicke weight \(q\) arising from such a loss branch, the integer
\(L(q)\) recovers the original coarse Dicke index \(\ell\).
This allows us to define recovery-free measurements directly on the surviving \(m\)-particle
block.
The first measurement is the projective observable
\begin{equation}
Z_g^{(m)}
=
\sum_{q=0}^{m}(-1)^{L(q)}
\ket{D_q^m}\bra{D_q^m}.
\label{eq:Zg_main}
\end{equation}
It reads the parity of the reconstructed coarse index \(L(q)\), not the parity
of the surviving excitation number \(q\). Since \(L(q)=\ell\) on every valid
loss branch with \(a<g\), the logical value is unchanged even when an odd number
of excitations is lost. The
complementary measurement compares the surviving Dicke components that differ
by \(g\):
\begin{equation}
X_g^{(m)}
=
\sum_{\substack{0\le q\le m-g\\ L(q)\ {\rm even}}}
\left(
\ket{D_q^m}\bra{D_{q+g}^m}
+
\ket{D_{q+g}^m}\bra{D_q^m}
\right).
\label{eq:Xg_main}
\end{equation}
It asks whether each such pair has positive or negative relative phase. On
Dicke components not belonging to a displayed pair, \(X_g^{(m)}\) vanishes:
these components carry no useful answer to this coherence question. Since
\(\|X_g^{(m)}\|\le1\), it nevertheless defines a valid binary POVM, with
effects \((I\pm X_g^{(m)})/2\).
Writing \(Z_g=\bigoplus_m Z_g^{(m)}\) and
\(X_g=\bigoplus_m X_g^{(m)}\), Alice and Bob use the binary observables
\begin{equation}
A_0=Z_g,\quad A_1=X_g,
\qquad
B_{0,1}=\frac{Z_g\pm X_g}{\sqrt2}.
\end{equation}
These are well-defined because \(Z_g^{(m)}\) and \(X_g^{(m)}\) anticommute on
every paired subspace, while \(X_g^{(m)}=0\) elsewhere; hence
\(\|B_0\|,\|B_1\|\le1\). Bob's two settings combine the logical-value question
with the coherence question, with opposite signs for the latter. No syndrome
extraction, decoding, or recovery operation is applied.

The full finite-size bound is proved in Supplemental Material,
Secs.~S5. One explicit consequence follows by choosing an odd integer
\(\mathsf M\ge101\) and setting
$g=\mathsf M$, $t=\mathsf M-1$, $n=\mathsf M^2$.
This recovery-free protocol satisfies \(S>2\) whenever the particle survival probability obeys $\eta\ge1-\frac{1}{4 \,r}$, for $r\sim \mathcal{O}(\sqrt n)$ deleted particles.
In this family, adding more particles allows the protocol to withstand more
deletions: for \(n\) particles, it tolerates up to about \(\sqrt n\) deleted
particles. However, this does not mean that it tolerates a fixed fraction of
lost particles. The loss probability per particle must decrease as
\(1/\sqrt n\), so that the typical total number of deletions remains of order
\(r\sim\sqrt n\). This is therefore a sparse-loss redundancy result. It differs
from the low-excitation protocols discussed above, which can violate CHSH even when each particle has the same fixed survival probability \(\eta<1\) as
\(n\) increases.
%

Together, these protocols reveal complementary routes to recovery-free nonlocality: low-excitation encodings retain coherence at a constant loss rate, whereas the binomial construction uses redundancy to tolerate a growing number of deletions.
Sections~S4.2 and S4.3 of the Supplemental Material analyse imperfect-measurement models for guiding an experimental proposal.

\section{Experimental feasibility}
All protocols measure the post-loss state directly with different
coherent controls. The one-excitation and \(N\)-excitation schemes require a binary measurement that resolves and coherently mixes the vacuum with one chosen excitation sector. The binomial protocol instead requires coherent measurements between several pairs of Dicke states whose excitation numbers differ by \(g\).

The one-excitation protocol is therefore the most accessible route to a first recovery-free Bell test. It requires vacuum--one-excitation discrimination and control of their relative phase. Circuit and cavity QED offer number-selective
control and readout in this single-rail subspace~\cite{Heeres2015,Krastanov2015,Hu2019}. Photonics provides a complementary route through displacement-assisted photon counting or homodyne detection~\cite{Banaszek1999,Hessmo2004}. The higher-excitation and binomial protocols are more demanding: the former requires coherent control between \(\ket{0}\) and
\(\ket{N}\), while the latter requires this control across many Dicke-state pairs. Trapped-ion chains and multiphoton platforms have demonstrated the preparation of symmetric Dicke states and are natural candidates for exploring these more redundant encodings \cite{Hume2009,Prevedel2009}.

The central experimental limitation is imperfect phase-sensitive
measurement. In contrast, the vacuum/nonvacuum measurement can remain reliable after loss. We describe this imbalance by a contrast \(v\leq1\) that reduces the phase-sensitive component of a local measurement while leaving the population measurement unchanged. A CHSH violation then requires a
larger survival probability than in the ideal case. 
%
For the $N$-excitation family, a fixed nonzero contrast changes the required excitation number but preserves the asymptotic inverse-golden-ratio boundary (see also Supp. Mat. S6). 

\section{Conclusions}
Recovery-free CHSH nonlocality has no universal loss threshold: it depends on the measurements allowed on the post-loss state. With arbitrary measurements, any deterministic recovery can be absorbed into the final measurement. For flagged erasure, the sharp limit is \(\eta=1/2\).
For explicit protocols without recovery, the protocol closer to $\eta=1/2$ from above requires to increase the number of excitation and reaches \(\etaG=(\sqrt5-1)/2\simeq 0.618\).
A further route is to optimize the logical input state. We used $\ket{\Phi^+_L}$, but lossy Bell tests can benefit from an unbalanced state $\ket{\psi_{\lambda}}
=
\sqrt{\lambda}\ket{0_L0_L}
+
\sqrt{1-\lambda}\ket{1_L1_L}$: it reduces the loss-sensitive part while preserving the coherence needed for nonlocality. This is the idea behind Eberhard's strategy~\cite{Eberhard1993,Gigena2025}. 
The open question is whether these Eberhard-like encodings and other
measurements, can violate a Bell inequality for
$\frac12<\eta\le\eta_G$,
or whether the golden-ratio threshold remains a limit for
experimentally accessible measurements.

\section*{acknowledgments}
This work has received funding from the
European Union's Horizon Europe research and innovation programme under the Project ``Quantum Secure Networks Partnership'' (QSNP, Grant Agreement No.~101114043);
and from INFN (Istituto Nazionale di Fisica Nucleare) through the project “QUANTUM.”
OpenAI ChatGPT, using the GPT-5.6 Sol model with ultra reasoning mode and Lean4 formal-verification scripts, were used as assisting tools during manuscript preparation and computational checking\cite{scala_logic_BI}.
\bibliographystyle{apsrev4-2}
\bibliography{logic_BI}
\appendix
\setcounter{section}{0}
\renewcommand{\thesection}{S\arabic{section}}
\setcounter{equation}{0}
\numberwithin{equation}{section}

\begin{widetext}
\section{Operational models}
\label{supp:channels}

In this section, we fix the notation and rigorously support the conceptual results  with the technical details that answer our central question: \emph{how much CHSH nonlocality remains accessible when one operationally explicitly measures the damaged state directly, without first reconstructing the encoded state?}
%
Let us recall the symmetric spaces in which the PI encodings are defined.

\begin{definition}[Dicke states and symmetric survivor spaces]
\label{def:supp-dicke-state}
For \(0\le k\le n\), the \(n\)-qubit Dicke state with \(k\) excitations is
\begin{equation}
\ket{D_k^n}
=
\binom nk^{-1/2}
\sum_{x\in\{0,1\}^n:\, |x|=k}
\ket{x}.
\label{eq:supp-dicke-state}
\end{equation}
The symmetric subspace is
$\cH_{n,\sym}
=
\operatorname{span}\{\ket{D_k^n}:0\le k\le n\}$.
After unflagged deletion, a branch with \(m\) surviving particles lies in
\(\cH_{m,\sym}\).  Since the number of surviving particles can vary from run to run, the full PI output space is
$\cH_{\rm surv}^{\rm PI}
:=
\bigoplus_{m=0}^n \cH_{m,\sym}$.

\noindent For a non-symmetric system of distinguishable particles, the order of the surviving particles matters.  In that case, if \(m\) particles survive, the output lives in \(\cH^{\otimes m}\), and the full survivor-number space is
$\bigoplus_{m=0}^n \cH^{\otimes m}$.
\end{definition}

Let us recall in lemma \ref{lem:supp-dicke-split} the decomposition that separates the original Dicke state into a survivor block of \(n-r\) particles and a lost block of \(r\) particles.  In each term, \(\ket{D_{k-a}^{n-r}}\) describes the survivors, while \(\ket{D_a^r}\) describes the particles that are deleted.  The exact deletion channel is then obtained by tracing out this lost block.

\begin{lemma}[Dicke split]
\label{lem:supp-dicke-split}
Let \(0\le r\le n\).  Splitting \(n\) carriers into \(n-r\) survivors and \(r\) deleted carriers gives
\begin{equation}
\ket{D_k^n}
=
\sum_{a=0}^{r}
\left[
\frac{\binom{r}{a}\binom{n-r}{k-a}}
     {\binom nk}
\right]^{1/2}
\ket{D_{k-a}^{n-r}}\ket{D_a^r},
\label{eq:supp-dicke-split}
\end{equation}
where binomial coefficients outside their natural range are zero.
\end{lemma}

\begin{proof}
Expand \(\ket{D_k^n}\) in the computational basis and group the strings according to the number \(a\) of excitations in the deleted block.  For a fixed \(a\), there are
\(\binom{n-r}{k-a}\binom r a\) such strings.  Their normalized equal superposition factorizes as
\(\ket{D_{k-a}^{n-r}}\ket{D_a^r}\).  Comparing its normalization with the normalization
\(\binom nk^{-1/2}\) of \(\ket{D_k^n}\) gives the coefficient in Eq.~\eqref{eq:supp-dicke-split}.
\end{proof}

This decomposition is used on the loss models in the Bell experiment. The state prepared before loss is
\begin{equation}
\ket{\Phi_L}
=
\frac{
\ket{0_L}_A\ket{0_L}_B+
\ket{1_L}_A\ket{1_L}_B
}{\sqrt2},
\qquad
\rho_L=\ket{\Phi_L}\!\bra{\Phi_L},
\label{eq:supp-initial-logical-bell-state}
\end{equation}
where each logical state belongs to \(\cH_{n,\sym}\).
Let \(\cH\) be the Hilbert space of one particle and define
$\cH_\varnothing
:=
\cH\oplus\operatorname{span}\{\ket\varnothing\}$,
 with
$\ket\varnothing\perp\cH$.
For the permutationally invariant PI protocols, \(\cH=\mathbb C^2\). 

The single-particle flagged-erasure channel is
\begin{equation}
\cE_\eta:\cL(\cH)\to\cL(\cH_\varnothing),
\qquad
\cE_\eta(\rho)
=
\eta\rho+
(1-\eta)\Tr(\rho)\ket\varnothing\!\bra\varnothing.
\label{eq:supp-flagged-erasure-channel}
\end{equation}
Applied independently to the \(n\) particles in each laboratory, it gives the flagged output
\begin{equation}
\rho_{AB}^{\rm flag}
=
\bigl(
\cE_\eta^{\otimes n}\otimes
\cE_\eta^{\otimes n}
\bigr)(\rho_L).
\label{eq:supp-flagged-bell-output}
\end{equation}
This state records which physical particles were lost in Alice's and Bob's registers.  

Unflagged deletion removes this positional information.  
Let \([n]=\{1,\ldots,n\}\), and consider a survivor set
\(S=\{i_1<\cdots<i_m\}\subseteq[n]\).  Let \(Q_S\) be the projector onto the flag pattern in which precisely the particles in \(S\) survive.  On this branch, define the partial isometry
\begin{equation}
W_S:
Q_S\cH_\varnothing^{\otimes n}
\longrightarrow
\cH^{\otimes m}
\subset
\bigoplus_{j=0}^n\cH^{\otimes j},
\qquad
W_S\!\left(
\ket{\psi_1}_{i_1}\cdots\ket{\psi_m}_{i_m}
\otimes
\ket\varnothing_{[n]\setminus S}
\right)
=
\ket{\psi_1}\otimes\cdots\otimes\ket{\psi_m},
\label{eq:supp-packing-map}
\end{equation}
and extend it linearly.  Thus \(W_S\) removes the empty slots and preserves the order of the surviving particles.  Since
\begin{equation}
W_S^\dagger W_S=Q_S,
\qquad
\sum_{S\subseteq[n]}Q_S=I_{\cH_\varnothing^{\otimes n}},
\end{equation}
the packing channel
\begin{equation}
\cP_{\rm pack}(\omega)
=
\sum_{S\subseteq[n]}W_S\,\omega\,W_S^\dagger
\label{eq:supp-packing-channel}
\end{equation}
is completely positive and trace preserving.
The one-party unflagged deletion channel is therefore
$\cD_{n,\eta}^{\rm unflag}
=
\cP_{\rm pack}\circ\cE_\eta^{\otimes n}$.
Accordingly, the state that reaches Alice's and Bob's laboratories after unflagged particle loss is
\begin{align}
\rho_{AB}^{\rm unflag}
=
\bigl(
\cD_{n,\eta}^{\rm unflag}
\otimes
\cD_{n,\eta}^{\rm unflag}
\bigr)(\rho_L)
=
\bigl(
\cP_{\rm pack}\otimes\cP_{\rm pack}
\bigr)
\bigl(\rho_{AB}^{\rm flag}\bigr).
\label{eq:supp-unflagged-bell-output}
\end{align}
Packing is only a local processing of the output state.  It removes the information about which particles were lost, hence it does not increase entanglement.  Equivalently, every measurement that Alice or Bob performs after packing can also be regarded as a measurement on the flagged output, with the packing map included as part of the local measurement procedure.  Therefore, if the flagged output $\rho_{AB}^{\rm flag}$ is CHSH-local, then the unflagged output $\rho_{AB}^{\rm unflag}$ obtained from it by packing is also CHSH-local. This relation will be used in the no-go result of Sec.~\ref{supp:no-go}.  

We now specialize this description to PI states.  Before this specialization, the final state may depend on which particles survived, because the surviving particles keep their order in the shorter register.  For PI states this dependence disappears.  Since a PI state is invariant under permutations of the particles, all loss patterns with the same number \(r\) of deleted particles give the same reduced channel.  Thus the relevant information is only how many particles were deleted, not which ones.  If \(r\) particles are deleted, then \(m=n-r\) particles survive and the output lies in \(\cH_{m,\sym}\).  The Dicke split above gives the corresponding exact-\(r\) deletion channel.
\begin{lemma}[Exact-\(r\) PI deletion Kraus form]
\label{lem:supp-exact-deletion-kraus}
On \(\cH_{n,\sym}\), exact deletion of \(r\) particles is described by Kraus maps
\begin{equation}
K_{a}^{(r)}:\cH_{n,\sym}\to\cH_{n-r,\sym},
\qquad
K_a^{(r)}\ket{D_k^n}
=
\sqrt{q_{k,a}^{(r)}}\ket{D_{k-a}^{n-r}},
\qquad
q_{k,a}^{(r)}
=
\frac{\binom{k}{a}\binom{n-k}{r-a}}{\binom n r}.
\label{eq:supp-exact-deletion-kraus}
\end{equation}
Equivalently, the channel is
$\Del_n^{(r)}(\rho)
=
\sum_{a=0}^r
K_a^{(r)}\rho K_a^{(r)\dagger}$ ($n$ dependence is hidden in the domain of $K_a^{(r)}$).
Applying this expression to a Dicke matrix unit gives
\begin{align}
\Del_n^{(r)}
\!\left(
\ket{D_k^n}\!\bra{D_\ell^n}
\right)
=
\sum_{a=0}^r
K_a^{(r)}
\ket{D_k^n}\!\bra{D_\ell^n}
K_a^{(r)\dagger}
=
\sum_{a=0}^{r}
\sqrt{q_{k,a}^{(r)}q_{\ell,a}^{(r)}}
\ket{D_{k-a}^{n-r}}\!\bra{D_{\ell-a}^{n-r}}.
\label{eq:supp-exact-deletion-offdiag}
\end{align}
\end{lemma}

\begin{proof}
Apply Lemma~\ref{lem:supp-dicke-split} and trace out the deleted subsystem.  The deleted Dicke states are orthogonal, so only terms with the same number \(a\) of deleted excitations remain.  The coefficient in Eq.~\eqref{eq:supp-dicke-split} can be rewritten as
$\binom r a\binom{n-r}{k-a}\big/\binom n k
=
\binom k a\binom{n-k}{r-a}\big/\binom n r$.
This gives Eqs.~\eqref{eq:supp-exact-deletion-kraus}--\eqref{eq:supp-exact-deletion-offdiag}.  For fixed \(k\),
\(\sum_a q_{k,a}^{(r)}=1\), so the Kraus operators are complete.
\end{proof}

So far, the deletion number \(r\) has been fixed.  Under independent particle loss, it fluctuates from run to run.  The full PI channel is obtained by combining the exact-\(r\) channels with their binomial probabilities.  We keep the different survivor-number sectors as a direct sum because the number \(m=n-r\) is available from the length of the packed output register.

\begin{definition}[Stochastic PI deletion]
\label{def:supp-stochastic-pi-deletion}
The stochastic independent-deletion channel on the PI subspace is
\begin{equation}
\Del_{n,\eta}^{\rm PI}
=
\bigoplus_{r=0}^n
\Pr[\Rdel=r]\Del_n^{(r)},
\qquad
\Pr[\Rdel=r]
=
\binom nr (1-\eta)^r\eta^{n-r},
\label{eq:supp-stochastic-pi-deletion}
\end{equation}
where \(\Rdel\sim\Bin(n,1-\eta)\) and
\(\Jsurv=n-\Rdel\sim\Bin(n,\eta)\).  The branch with \(r\) deletions belongs to the survivor sector with \(m=n-r\) particles.
\end{definition}
We have now clarified the two loss models: flagged and unflagged particle loss.
For the sake of the rigor, we next clarify how deterministic recovery should be interpreted in a Bell test.  
%
%

Let
\(
\cR:\cL(\mathcal K)\to\cL(\mathcal H)
\)
be a deterministic quantum channel, where \(\mathcal K\) is the post-loss output space and \(\mathcal H\) is the recovered logical space.  Deterministic means completely positive and trace preserving: no outcome is selected and no run is discarded.  The adjoint map \(\cR^\dagger:\cL(\mathcal H)\to\cL(\mathcal K)\) is defined by
\begin{equation}
    \Tr\!\left[X\,\cR(\rho)\right]
=
\Tr\!\left[\cR^\dagger(X)\rho\right]
\label{eq:supp-adjoint-definition}
\end{equation}
for all states \(\rho\) on \(\mathcal K\) and all operators \(X\) on \(\mathcal H\).  This is the Hilbert--Schmidt adjoint, because it is the adjoint with respect to the inner product
\(
\langle X,Y\rangle_{\rm HS}=\Tr(X^\dagger Y)
\). Hence
\begin{equation}
\cR(\rho)=\sum_\alpha R_\alpha\rho R_\alpha^\dagger
\quad \Longrightarrow \quad
\cR^\dagger(X)
=
\sum_\alpha R_\alpha^\dagger X R_\alpha.
\label{eq:supp-adjoint-kraus}
\end{equation}

\begin{proposition}[Deterministic recovery can be included in the measurement]
\label{prop:supp-recovery-as-measurement}
Let \(\rho_{\rm loss}\) be any bipartite post-loss state on
\(\mathcal K_A\otimes\mathcal K_B\).  Let
\(
\cR_A:\cL(\mathcal K_A)\to\cL(\mathcal H_A),
\) and \(
\cR_B:\cL(\mathcal K_B)\to\cL(\mathcal H_B)
\)
be deterministic local recovery channels.  Suppose that Alice and Bob then perform POVMs
\(\{M^A_{a|x}\}_a\) and \(\{M^B_{b|y}\}_b\) on the recovered logical spaces.  Define
$\widetilde M^A_{a|x}
=
\cR_A^\dagger(M^A_{a|x})$,
and
$\widetilde M^B_{b|y}
=
\cR_B^\dagger(M^B_{b|y})$.
Then \(\{\widetilde M^A_{a|x}\}_a\) and
\(\{\widetilde M^B_{b|y}\}_b\) are valid POVMs on the post-loss spaces, and
\begin{align}\Tr\!\left[
(M^A_{a|x}\otimes M^B_{b|y})
(\cR_A\otimes\cR_B)(\rho_{\rm loss})
\right]
=
\Tr\!\left[
(\widetilde M^A_{a|x}\otimes \widetilde M^B_{b|y})
\rho_{\rm loss}
\right].
\end{align}
Therefore the CHSH value is unchanged.
\end{proposition}
\begin{proof}
Simply, since \(\cR\) is trace preserving, its Kraus operators satisfy
\(
\sum_\alpha R_\alpha^\dagger R_\alpha=I_{\mathcal K}
\). Therefore Eq.~\eqref{eq:supp-adjoint-kraus} gives
\(
\cR^\dagger(I_{\mathcal H})=I_{\mathcal K}
\). Moreover, if \(X\ge0\), then \(\cR^\dagger(X)\ge0\).  Hence a POVM
\(\{M_b\}_b\), with \(M_b\ge0\) and \(\sum_b M_b=I_{\mathcal H}\), is sent to another POVM:
\(
\cR^\dagger(M_b)\ge0\),
 \(\sum_b\cR^\dagger(M_b)=I_{\mathcal K}
\). The equality of probabilities follows directly from the defining identity
\eqref{eq:supp-adjoint-definition}, applied once on Alice's side and once on Bob's side.

For a binary measurement, the signed observable
\(
A=M_+-M_-
\)
is transformed into
\(
\widetilde A=\cR^\dagger(A)
\). It is Hermitian, and
\(
I_{\mathcal K}\pm\widetilde A
=
2\cR^\dagger(M_\pm)\ge0.
\)
Thus \(-I_{\mathcal K}\le\widetilde A\le I_{\mathcal K}\).  The effective observable is therefore a valid binary observable, although it need not be projective or experimentally simple.
\end{proof}
Going back to our central question, proposition~\ref{prop:supp-recovery-as-measurement} 
focuses on the meaning of measuring the damaged state ``without first reconstructing the encoded state''.  At the level of measurement statistics, a deterministic recovery followed by a measurement is equivalent to a single local POVM performed directly on the post-loss state.  This observation does not, however, make the problem easier: the effective POVM may be just as difficult to implement as the recovery itself.  Moreover, the proposition neither constructs a recovery channel nor tells us how much Bell nonlocality survives the loss.  Not only that, answering our central question still requires us to examine what information remains available in each loss model and whether reliable recovery is guaranteed.  
%
%

\section{Universal no-go for \(\eta\le1/2\)}
\label{supp:no-go}

This section focuses on how much Bell nonlocality can remain after loss, independently of the encoding and of the chosen measurements.  We first prove in detail that, when the survival probability satisfies \(\eta\le1/2\), CHSH violation is impossible even if Alice and Bob are allowed arbitrary local POVMs on the post-loss state as in Proposition~\ref{prop:supp-recovery-as-measurement}.  The proof is given for flagged erasure and then transferred to unflagged deletion using the local packing relation established in Sec.~\ref{supp:channels}.  

For \(\eta>1/2\), the situation changes.  The flagged-erasure channel has positive quantum capacity, which allows arbitrary local POVMs on the flagged outputs to attain CHSH values approaching asymptotically \(2\sqrt{2}\).  However, we do not specify the operational measurements. Still, it remains unclear whether $\eta=1/2$ is the lowest threshold at which accessible nonlocality persists for unflagged deletion.

The origin of the boundary \(\eta=1/2\) is most transparent at the level of a single use of the channel.  When \(\eta\leq 1/2\), the input can be routed (\emph{symmetric extension}) to either of two receiver registers with probability \(\eta\), while both registers receive an erasure flag with the remaining probability.  

\begin{lemma}[Symmetric extension of the flagged-erasure channel]
\label{lem:supp-flagged-erasure-two-extension}
Take $\mathcal{H}_{\varnothing}
=
\mathcal{H}\oplus\operatorname{span}\!\left\{\ket{\varnothing}\right\}$,
where \(\ket{\varnothing}\perp\mathcal{H}\).  For
\(0\leq\eta\leq 1/2\), the single-particle flagged-erasure channel of Eq.~\eqref{eq:supp-flagged-erasure-channel}
admits the symmetric two-output extension
\begin{align}
\widetilde{\mathcal{E}}_{\eta}^{(2)}(\rho)
=
\eta\,\rho_{B_1}\otimes
\ket{\varnothing}\!\bra{\varnothing}_{B_2}
+
\eta\,\ket{\varnothing}\!\bra{\varnothing}_{B_1}
\otimes\rho_{B_2}
+
(1-2\eta)\operatorname{Tr}(\rho)
\ket{\varnothing\varnothing}\!
\bra{\varnothing\varnothing}_{B_1B_2}.
\label{eq:supp-two-extension-channel}
\end{align}
The joint output is invariant under the exchange
\(B_1\leftrightarrow B_2\), and each marginal equals
\(\mathcal{E}_{\eta}(\rho)\).
\end{lemma}

\begin{proof}
The three coefficients in Eq.~\eqref{eq:supp-two-extension-channel} are nonnegative when \(\eta\le1/2\).  The map is therefore completely positive.  It is also trace preserving: $\eta\Tr\rho+\eta\Tr\rho+(1-2\eta)\Tr\rho=\Tr\rho$. The formula is symmetric under \(B_1\leftrightarrow B_2\).  Taking the partial trace over \(B_2\) (identical on \(B_1\)) gives
\begin{align}
\Tr_{B_2}\widetilde\cE_\eta^{(2)}(\rho)
&=\eta\rho_{B_1}
+\eta\Tr(\rho)\ket\varnothing\!\bra\varnothing_{B_1}
+(1-2\eta)\Tr(\rho)\ket\varnothing\!\bra\varnothing_{B_1}
\nonumber\\
&=\eta\rho_{B_1}+(1-\eta)\Tr(\rho)\ket\varnothing\!\bra\varnothing_{B_1}
=\cE_\eta(\rho).
\end{align}
\end{proof}
For \(n\) independently transmitted particles, the extension is
\((\widetilde\cE_\eta^{(2)})^{\otimes n}\).  It is symmetric under exchange of the two complete receiver registers, and both marginals equal \(\cE_\eta^{\otimes n}\).

The channel extension must now be connected to Bell locality.  Because CHSH
contains only two settings for Bob, a two-extension allows those two
measurements to be represented simultaneously on different copies of Bob's
register.  Their outcomes then form a classical hidden variable, as made
explicit below.

\begin{lemma}[Locality of two-extendible states for two Bob settings]
\label{lem:supp-two-extendible-two-setting-local}
Let \(\rho_{AB}\) admit a symmetric two-extension on Bob's side.  Then
\(\rho_{AB}\) has a local hidden-variable model for every Bell experiment in
which Bob has at most two measurement settings.  In particular,
\(\rho_{AB}\) cannot violate CHSH.
\end{lemma}

\begin{proof} Let \(\omega_{AB_0B_1}\) be a two-extension of \(\rho_{AB}\), so that its two Alice--Bob marginals satisfy \( \omega_{AB_0}=\omega_{AB_1}=\rho_{AB}, \) up to the natural identification of \(B_0\) and \(B_1\) with Bob's original system. Let \(\{M_{a|x}\}_a\) be Alice's POVM for setting \(x\), and let \(\{N_{b|0}\}_b\) and \(\{N_{b|1}\}_b\) be Bob's two POVMs. Since \(B_0\) and \(B_1\) are distinct systems, Bob's setting-\(0\) measurement can be applied to \(B_0\) and Bob's setting-\(1\) measurement to \(B_1\) simultaneously. Together with Alice's measurement, this defines the joint distribution 
\begin{equation} 
P(a,b_0,b_1|x) = \Tr\!\left[ \bigl( M_{a|x}\otimes N_{b_0|0}\otimes N_{b_1|1} \bigr) \omega_{AB_0B_1} \right]. 
\end{equation} 
After summing over Alice's outcome, we obtain $
P(b_0,b_1) = \sum_a P(a,b_0,b_1|x)$.  
Because \(\sum_a M_{a|x}=I_A\), this marginal is independent of Alice's setting \(x\). We now take the hidden variable to be the pair \(\lambda=(b_0,b_1)\).  Operationally, \(\lambda\) stores one potential output for each of Bob's two settings: \(b_0\) for setting \(0\) and \(b_1\) for setting \(1\). Its probability distribution is \( P(\lambda)=P(b_0,b_1)\). For every \(\lambda\) with \(P(\lambda)>0\), define Alice and Bob's local response by \begin{equation} 
P_A(a|x,\lambda) = \frac{P(a,b_0,b_1|x)}{P(b_0,b_1)},
\qquad
P_B(b|y,\lambda) = \delta_{b,\lambda_y} = 
\begin{cases} 
\delta_{b,b_0}, & y=0,\\ \delta_{b,b_1}, & y=1, 
\end{cases} 
\label{eq:supp-two-setting-lhv-model}
\end{equation} 
For zero-probability values of \(\lambda\), \(P_A\) may be defined arbitrarily. Bob's response is deterministic once \(y\) and \(\lambda\) are specified. \(\lambda_y=b_y\) denotes the component of \(\lambda\) corresponding to Bob's chosen setting. For \(y=0\), the local model gives 
\begin{align} 
\sum_\lambda P(\lambda)P_A(a|x,\lambda)P_B(b|0,\lambda) &= \sum_{b_0,b_1} P(a,b_0,b_1|x)\delta_{b,b_0} = \sum_{b_1}P(a,b,b_1|x) \nonumber\\ &= \Tr\!\left[ (M_{a|x}\otimes N_{b|0})\omega_{AB_0} \right] = \Tr\!\left[ (M_{a|x}\otimes N_{b|0})\rho_{AB} \right]. \end{align} 
The same calculation using the \(AB_1\) marginal applies for \(y=1\).  The
local model therefore reproduces all probabilities in the Bell experiment,
and CHSH cannot be violated.
\end{proof}

This construction is the symmetric-extension version of the standard connection between shareability and local hidden-variable models \cite{Terhal2003}. For CHSH, one may equivalently use the usual monogamy relation \cite{Toner2006}.

\begin{theorem}[Universal CHSH no-go for \(\eta\leq 1/2\)]
\label{thm:supp-half-loss-no-go}
For \(\eta\leq 1/2\), deterministic recovery before measurement gives  $S\leq 2$ for flagged and unflagged deletions.
\end{theorem}

\begin{proof}
By Lemma~\ref{lem:supp-flagged-erasure-two-extension}, each flagged-erasure channel on Bob's side has a symmetric two-output extension when \(\eta\leq 1/2\).  Applying the tensor product of these extensions to an arbitrary bipartite input produces a symmetric extension of the complete post-loss state on Bob's side.
Lemma~\ref{lem:supp-two-extendible-two-setting-local} therefore implies that all CHSH correlations obtained from the flagged output are local.

The unflagged output is obtained by applying the packing channel
\(\mathcal{P}_{\mathrm{pack}}\) of Eq.~\eqref{eq:supp-packing-channel} locally to the flagged output.  Any measurement performed after packing can also be regarded as a measurement on the flagged state, with packing included in the local measurement procedure as shown in Proposition~\ref{prop:supp-recovery-as-measurement}.
Packing therefore cannot turn a CHSH-local flagged state into a CHSH-nonlocal unflagged state.  Additional loss or other local processing on Alice's side cannot restore a violation.
\end{proof}

For \(\eta>1/2\), we now show how the quantum coding theorem for flagged erasure provides encodings and deterministic decoders whose effective logical channels approach the identity.  The following continuity estimate translates this channel-fidelity statement into a CHSH statement.

We use the trace distance
$D(\rho,\tau)
:=
\lVert\rho-\tau\rVert_1/2$,
which controls how much measurement statistics can change when one state is
replaced by another.

\begin{lemma}[Bell-state fidelity guarantees a large CHSH value]
\label{lem:supp-chsh-continuity}
Let
$\Phi=\ket{\Phi^+}\!\bra{\Phi^+}$
be a two-qubit Bell state, and let \(\sigma\) be any two-qubit state.  Fix a
CHSH operator \(\mathcal{B}_{\Phi}\) corresponding to measurements that attain
the Tsirelson value on \(\Phi\):
$\operatorname{Tr}(\mathcal{B}_{\Phi}\Phi)=2\sqrt{2}$.
If $F(\sigma,\Phi):=\bra{\Phi^+}\sigma\ket{\Phi^+}
\geq 1-\epsilon$, then the same measurements give
$\operatorname{Tr}(\mathcal{B}_{\Phi}\sigma)
\geq
2\sqrt{2}-4\sqrt{2\epsilon}$.
\end{lemma}

\begin{proof}
Because \(\Phi\) is pure, the Fuchs--van de Graaf inequality gives
$D(\sigma,\Phi)
\leq
\sqrt{1-F(\sigma,\Phi)}
\leq
\sqrt{\epsilon}$.
Equivalently, \(\lVert\sigma-\Phi\rVert_1\leq 2\sqrt{\epsilon}\).  The change
in the expectation value of \(\mathcal{B}_{\Phi}\) satisfies
$\left|
\operatorname{Tr}\!\left[
\mathcal{B}_{\Phi}(\sigma-\Phi)
\right]
\right|
\leq
\lVert\mathcal{B}_{\Phi}\rVert_{\infty}
\lVert\sigma-\Phi\rVert_1$.
For a CHSH operator,
\(\lVert\mathcal{B}_{\Phi}\rVert_{\infty}\leq 2\sqrt{2}\).  It follows that
\begin{align}
\operatorname{Tr}(\mathcal{B}_{\Phi}\sigma)
=
\operatorname{Tr}(\mathcal{B}_{\Phi}\Phi)
+
\operatorname{Tr}\!\left[
\mathcal{B}_{\Phi}(\sigma-\Phi)
\right]
\geq
2\sqrt{2}-(2\sqrt{2})(2\sqrt{\epsilon})
=
2\sqrt{2}-4\sqrt{2\epsilon}.
\end{align}
The optimized CHSH value cannot be smaller than the value obtained with this
fixed choice of measurements.
\end{proof}

The capacity theorem controls the two local transmission procedures
separately.  The next corollary combines their errors and bounds the CHSH value
when both halves of the Bell state pass through their respective effective
logical channels.
\begin{corollary}[Combining the two local fidelity guarantees]
\label{cor:supp-two-sided-continuity}
Let \(\Lambda_A\) and \(\Lambda_B\) be deterministic channels acting on
Alice's and Bob's logical qubits, respectively.  Suppose that
$\operatorname{Tr}\!\left[
\Phi(\Lambda_A\otimes\operatorname{id})(\Phi)
\right]
\geq 1-\epsilon_A$,
$\operatorname{Tr}\!\left[
\Phi(\operatorname{id}\otimes\Lambda_B)(\Phi)
\right]
\geq 1-\epsilon_B$.
Then
$\sigma=(\Lambda_A\otimes\Lambda_B)(\Phi)$
satisfies
\begin{equation}
D(\sigma,\Phi)
\leq
\sqrt{\epsilon_A}+\sqrt{\epsilon_B}.
\label{eq:supp-two-sided-trace-distance}
\end{equation}
The Bell-state CHSH measurements therefore give
$\operatorname{Tr}(\mathcal{B}_{\Phi}\sigma)
\geq
2\sqrt{2}
-
4\sqrt{2}
\left(
\sqrt{\epsilon_A}+\sqrt{\epsilon_B}
\right)$.
In particular, if \(\epsilon_A,\epsilon_B\leq\epsilon\), then
\begin{equation}
\operatorname{Tr}(\mathcal{B}_{\Phi}\sigma)
\geq
2\sqrt{2}-8\sqrt{2\epsilon}.
\label{eq:supp-two-sided-chsh-equal}
\end{equation}
The optimized CHSH value satisfies the same lower bounds.
\end{corollary}

\begin{proof}
Define the one-sided outputs
$\sigma_A=(\Lambda_A\otimes\operatorname{id})(\Phi)$,
$\sigma_B=(\operatorname{id}\otimes\Lambda_B)(\Phi)$.
The two fidelity assumptions and the Fuchs--van de Graaf inequality imply
$D(\sigma_A,\Phi)\leq\sqrt{\epsilon_A}$,
$D(\sigma_B,\Phi)\leq\sqrt{\epsilon_B}$.
Since
$\sigma=(\Lambda_A\otimes\operatorname{id})(\sigma_B)$,
$\sigma_A=(\Lambda_A\otimes\operatorname{id})(\Phi)$,
contractivity of trace distance under deterministic channels gives
$D(\sigma,\sigma_A)
\leq
D(\sigma_B,\Phi)
\leq
\sqrt{\epsilon_B}$.
The triangle inequality then yields
$D(\sigma,\Phi)
\leq
D(\sigma,\sigma_A)+D(\sigma_A,\Phi)
\leq
\sqrt{\epsilon_A}+\sqrt{\epsilon_B}$,
which proves Eq.~\eqref{eq:supp-two-sided-trace-distance}.
Finally,
\begin{align}
\operatorname{Tr}(\mathcal{B}_{\Phi}\sigma)
\geq
\operatorname{Tr}(\mathcal{B}_{\Phi}\Phi)
-
\lVert\mathcal{B}_{\Phi}\rVert_{\infty}
\lVert\sigma-\Phi\rVert_1
\geq
2\sqrt{2}
-
4\sqrt{2}
\left(
\sqrt{\epsilon_A}+\sqrt{\epsilon_B}
\right),
\end{align}
where we used
\(\lVert\sigma-\Phi\rVert_1=2D(\sigma,\Phi)\) and
\(\lVert\mathcal{B}_{\Phi}\rVert_{\infty}\leq 2\sqrt{2}\).
\end{proof}

\begin{proposition}[Flagged-erasure for $\eta>1/2$]
\label{prop:supp-flagged-erasure-benchmark}
For every \(\eta>1/2\) and every \(\delta>0\), there exist a block length
\(n\), local encodings of the two halves of a Bell state into \(n\) uses of
the flagged-erasure channel, and local POVMs on the flagged post-loss outputs
such that
\begin{equation}
S\geq 2\sqrt{2}-\delta.
\label{eq:supp-erasure-achievability}
\end{equation}
The recovery need not be implemented as a separate physical step, although the
resulting POVMs may contain the full complexity of the decoder.
\end{proposition}

\begin{proof}
The quantum capacity of the qubit flagged-erasure channel is
$Q(\mathcal{E}_{\eta})
=
\max\{0,2\eta-1\}$.
It is therefore positive precisely when \(\eta>1/2\) \cite{Bennett1997}.  The
quantum channel coding theorem then provides, for a sufficiently large block
length \(n\), encoding channels \(\mathcal{V}_{A,n}\),
\(\mathcal{V}_{B,n}\) and deterministic recovery channels
\(\mathcal{R}_{A,n}\), \(\mathcal{R}_{B,n}\) such that the effective logical
channels
$\Lambda_{A,n}
=
\mathcal{R}_{A,n}\circ
\mathcal{E}_{\eta}^{\otimes n}\circ
\mathcal{V}_{A,n}$, $\Lambda_{B,n}
=
\mathcal{R}_{B,n}\circ
\mathcal{E}_{\eta}^{\otimes n}\circ
\mathcal{V}_{B,n}$
have entanglement fidelities arbitrarily close to one
\cite{Schumacher1996,Lloyd1997,Devetak2005,Barnum2000}.  In particular, the
codes can be chosen so that
$\operatorname{Tr}\!\left[
\Phi(\Lambda_{A,n}\otimes\operatorname{id})(\Phi)
\right]
\geq 1-\epsilon$,
$\operatorname{Tr}\!\left[
\Phi(\operatorname{id}\otimes\Lambda_{B,n})(\Phi)
\right]
\geq 1-\epsilon$.
Apply these encodings to the two halves of \(\Phi\), let both encoded blocks
pass through their local flagged-erasure channels, and then apply the two
recovery channels.  The resulting logical state is
$\sigma_n
=
(\Lambda_{A,n}\otimes\Lambda_{B,n})(\Phi)$.
Corollary~\ref{cor:supp-two-sided-continuity} shows that the Bell-state CHSH
measurements on \(\sigma_n\) give
$S(\sigma_n)
\geq
2\sqrt{2}-8\sqrt{2\epsilon}$.
Choosing
$\epsilon
\leq
\left(\frac{\delta}{8\sqrt{2}}\right)^2$
gives Eq.~\eqref{eq:supp-erasure-achievability}.

Finally, Proposition~\ref{prop:supp-recovery-as-measurement} includes each
deterministic recovery in the corresponding local measurement.  It therefore
produces valid POVMs on the flagged post-loss outputs with exactly the same
probabilities and the same CHSH value.
\end{proof}

Theorem~\ref{thm:supp-half-loss-no-go} and
Proposition~\ref{prop:supp-flagged-erasure-benchmark} characterize the
asymptotic flagged-erasure problem when arbitrary local POVMs are allowed:
\begin{align}
\eta\leq\frac{1}{2}
\quad\Longrightarrow\quad
S\leq 2;\;\;
\qquad
\eta>\frac{1}{2}
\quad\Longrightarrow\quad
S\longrightarrow 2\sqrt{2}
\quad
\text{for suitable block codes}.
\end{align}
The second implication is an asymptotic, unrestricted-measurement benchmark.
It does not determine the exact finite-block optimum or provide an
experimentally simple measurement.  The effective POVMs are defined through
the recovery channels and may be as difficult to implement as the recoveries
themselves. Going back to our central question: both the loss models does not have an operational procedure and does not exist CHSH violation for $\eta\le 1/2$. But for the flagged loss model, we showed that $\eta=1/2$ is the lowest possible threshold where nonlocality becomes unaccessible.

\paragraph{Why the flagged benchmark does not transfer to unflagged deletion.}

As shown in Sec.~\ref{supp:channels}, unflagged deletion is obtained from flagged erasure by discarding the erasure positions and packing the surviving particles:
\begin{equation}
\mathcal{D}_{n,\eta}^{\mathrm{unflag}}
=
\mathcal{P}_{\mathrm{pack}}
\circ
\mathcal{E}_{\eta}^{\otimes n}.
\label{eq:supp-deletion-as-postprocessed-erasure}
\end{equation}
Thus, every strategy for the unflagged output can also be applied to the flagged output by first ignoring the flags.  Consequently, the no-go result for flagged erasure transfers directly to unflagged deletion: discarding local information cannot create a CHSH violation.

The converse does not hold.  A flagged decoder may condition its action on the erasure pattern, which is unavailable after packing.  Such a decoder \(\mathcal{R}_{n}^{\mathrm{flag}}\) can be implemented on the unflagged output only if, on the relevant post-loss states, it factors as
\begin{equation}
\mathcal{R}_{n}^{\mathrm{flag}}
=
\mathcal{R}_{n}^{\mathrm{unflag}}
\circ
\mathcal{P}_{\mathrm{pack}}
\label{eq:supp-recovery-factorization}
\end{equation}
for some unflagged decoder \(\mathcal{R}_{n}^{\mathrm{unflag}}\).  Equivalently, after absorbing recovery into the measurement, each flagged effect must satisfy
\begin{equation}
\widetilde{M}_{b}^{\mathrm{flag}}
=
\mathcal{P}_{\mathrm{pack}}^{\dagger}
\left(M_{b}^{\mathrm{unflag}}\right)
\label{eq:supp-measurement-factorization}
\end{equation}
for some effect \(M_{b}^{\mathrm{unflag}}\) on the packed output.  The quantum-capacity theorem used in Proposition~\ref{prop:supp-flagged-erasure-benchmark} guarantees reliable decoding when the erasure positions are known, but it guarantees neither of these factorizations.  Positive capacity of the flagged channel therefore gives no achievability result for unflagged deletion.

Permutation-invariant encodings are naturally suited to this problem: after relabeling the survivors, their reduced state depends only on the number of deleted particles, not on their positions.  The general flagged-erasure capacity theorem, however, does not establish that capacity-achieving codes can be chosen permutation invariant or decoded using only the packed survivor state.

Hence the implication is strictly one-way: a code for unflagged deletion also works for flagged erasure, whereas a flagged code need not survive the removal of the flags.  Proposition~\ref{prop:supp-flagged-erasure-benchmark} therefore leaves the unflagged problem unresolved for \(\eta>1/2\).  This motivates the explicit PI encodings and survivor-sector measurements developed below, which produce recovery-free CHSH violations without access to the loss positions.

\section{One-excitation PI protocol}
\label{supp:one-excitation}

This section proves the finite-\(n\) one-excitation CHSH formula and its asymptotic threshold.  The input code is
$\ket{0_L}=\ket{D_0^n}$,
$\ket{1_L}=\ket{D_1^n}$.
For a survivor block of size \(m\), we identify
$\ket0_m:=\ket{D_0^m}$,
$\ket1_m:=\ket{D_1^m}$.
When no confusion can arise, the subscript \(m\) is omitted.

\begin{proposition}[One-excitation deletion as logical amplitude damping]
\label{prop:supp-one-ex-amplitude-damping}
Conditioned on the survivor number \(m=n-r\), exact PI deletion acts on the one-excitation logical matrix units as
\begin{align}
\ket0\!\bra0&\longmapsto \ket0_m\!\bra0_m,\quad
\ket1\!\bra1\longmapsto \frac{m}{n}\ket1_m\!\bra1_m+\left(1-\frac{m}{n}\right)\ket0_m\!\bra0_m,\quad
\ket0\!\bra1\longmapsto \sqrt{\frac{m}{n}}\,\ket0_m\!\bra1_m,
\label{eq:supp-one-ex-damping}
\end{align}
\end{proposition}

\begin{proof}
Use the exact-\(r\) PI deletion Kraus coefficients
$q_{k,a}^{(r)}=\frac{\binom{k}{a}\binom{n-k}{r-a}}{\binom nr}$.
For \(k=0\), only \(a=0\) contributes and \(q_{0,0}^{(r)}=1\), giving \(\ket0\!\bra0\mapsto\ket0_m\!\bra0_m\).  For \(k=1\), the two possible lost-excitation numbers are \(a=0,1\), with
\begin{equation}
q_{1,0}^{(r)}=\frac{\binom{n-1}{r}}{\binom nr}=\frac{n-r}{n}=\frac{m}{n},
\qquad
q_{1,1}^{(r)}=\frac{\binom{n-1}{r-1}}{\binom nr}=\frac{r}{n}=1-\frac{m}{n}.
\end{equation}
This gives the stated population map.  For the off-diagonal matrix unit, the channel action is a sum over the same lost-excitation number on ket and bra.  Since \(q_{0,a}^{(r)}\) is nonzero only for \(a=0\), one obtains the coefficient \(\sqrt{q_{0,0}^{(r)}q_{1,0}^{(r)}}=\sqrt{m/n}\).
\end{proof}
For each survivor block define the single-rail signed observables
\begin{align}
X_m=\ket0_m\!\bra1_m+\ket1_m\!\bra0_m,\qquad
Y_m=-i\ket0_m\!\bra1_m+i\ket1_m\!\bra0_m.
\label{eq:supp-one-ex-observables}
\end{align}
The post-deletion state has support only on this two-dimensional subspace in every survivor block. We set the signed observables to zero on its orthogonal complement.

\begin{lemma}[Fixed-block correlators]
\label{lem:supp-one-ex-fixed-correlators}
Let Alice and Bob have survivor numbers \(m_A,m_B\).  Set \(\eta_A=m_A/n\) and \(\eta_B=m_B/n\).  For the locally deleted encoded Bell state
$\ket{\Phi_L}=\frac{\ket{0_L0_L}+\ket{1_L1_L}}{\sqrt2},$
the blockwise observables in Eq.~\eqref{eq:supp-one-ex-observables} satisfy
\begin{equation}
\langle X\otimes X\rangle=\sqrt{\eta_A\eta_B},
\qquad
\langle Y\otimes Y\rangle=-\sqrt{\eta_A\eta_B},
\qquad
\langle X\otimes Y\rangle=\langle Y\otimes X\rangle=0.
\label{eq:supp-one-ex-fixed-correlators}
\end{equation}
Consequently the equatorial CHSH settings
\begin{equation}
A_0=X,
\quad
A_1=Y,
\qquad
B_0=\frac{X-Y}{\sqrt2},
\quad
B_1=\frac{X+Y}{\sqrt2}
\end{equation}
give
\begin{equation}
S(m_A,m_B)=2\sqrt2\,\frac{\sqrt{m_A m_B}}{n}.
\label{eq:supp-one-ex-fixed-S}
\end{equation}
\end{lemma}

\begin{proof}
The only matrix elements contributing to \(X\otimes X\) and \(Y\otimes Y\) are the coherences between \(\ket{0_m0_m}\) and \(\ket{1_m1_m}\).  By Proposition~\ref{prop:supp-one-ex-amplitude-damping}, the two-sided coherence of the encoded Bell state is multiplied by \(\sqrt{\eta_A\eta_B}\).  The Bell state convention gives \(\langle X\otimes X\rangle=1\) and \(\langle Y\otimes Y\rangle=-1\) before damping, hence Eq.~\eqref{eq:supp-one-ex-fixed-correlators}.  Mixed equatorial correlators vanish by the same phase convention.  Substituting these values in
\begin{equation}
A_0\otimes(B_0+B_1)+A_1\otimes(B_0-B_1)
=\sqrt2\,X\otimes X-\sqrt2\,Y\otimes Y
\end{equation}
gives Eq.~\eqref{eq:supp-one-ex-fixed-S}.
\end{proof}

\begin{theorem}[One-excitation CHSH value]
\label{thm:supp-one-ex-chsh}
For independent deletion probability \(p=1-\eta\), let \(\Jsurv\sim\Bin(n,\eta)\) be the survivor number on one side.  Direct blockwise single-rail measurements give
\begin{equation}
S_n(p)=2\sqrt2\left(\E\sqrt{\Jsurv/n}\right)^2.
\label{eq:supp-one-ex-finite-n}
\end{equation}
Moreover,
$S_n(p)\longrightarrow 2\sqrt2\,\eta$
as \(n\to\infty\).  Therefore, the ideal asymptotic one-excitation protocol violates CHSH iff
$\eta>\frac1{\sqrt2}$.
\end{theorem}
\begin{proof} Let \(M_{n,A},M_{n,B}\sim \operatorname{Bin}(n,\eta)\) be the independent survivor numbers on Alice's and Bob's sides. Averaging Eq.~\eqref{eq:supp-one-ex-fixed-S} over the two independent deletion processes gives 
\begin{align} 
S_n(p) = 2\sqrt2\, \E\!\left[ \sqrt{\frac{M_{n,A}}{n}} \sqrt{\frac{M_{n,B}}{n}} \right] 
= 2\sqrt2\, \E\sqrt{\frac{M_{n,A}}{n}}\, \E\sqrt{\frac{M_{n,B}}{n}} 
= 2\sqrt2 \left( \E\sqrt{\frac{M_n}{n}} \right)^2. 
\end{align} 
To evaluate the large-\(n\) limit, write \( M_n=\sum_{i=1}^n X_i, \) where the \(X_i\) are independent Bernoulli random variables with \(\Pr(X_i=1)=\eta\). Thus \(M_n/n\) is the fraction of particles that survive. By the law of large numbers, 
\[ 
\frac{M_n}{n}\xrightarrow{P}\eta \qquad\text{in probability, i.e. } \forall\delta>0\quad
\Pr\!\left( \left|\frac{M_n}{n}-\eta\right|>\delta \right)\longrightarrow0. \] 
Since \(x\mapsto\sqrt{x}\) is continuous on \([0,1]\), it follows that $\sqrt{\frac{M_n}{n}} \xrightarrow{P}\sqrt{\eta}$. We now pass from convergence in probability to convergence of the expectation. The variables \(\sqrt{M_n/n}\) lies in $[ 0,1 ]$ for every \(n\). This uniform bound implies uniform integrability: values in the tails cannot make a finite contribution to the expectation because the variables never exceed one. Convergence in probability together with uniform integrability therefore gives \( \E\left| \sqrt{\frac{M_n}{n}}-\sqrt{\eta} \right| \longrightarrow0, \) and hence \( \E\sqrt{\frac{M_n}{n}}\longrightarrow\sqrt{\eta}\). Substituting this limit into the finite-\(n\) expression yields \( S_n(p) \longrightarrow 2\sqrt2\,\eta\). The limiting value exceeds the local CHSH bound \(2\) precisely when \( \eta>\frac1{\sqrt2}\). \end{proof}
\section{\(N\)-excitation PI protocols and the pure-loss limit}
\label{supp:pure-loss-limit}
We consider PI deletion with finite $N$ excitations in the limit of particle $n\to\infty$.  Since deletion produces sectors with different numbers of survivors, we first identify their low-excitation states with states in a common finite-Fock space. We then prove that the resulting packed channel converges in diamond norm to the bosonic pure-loss channel. This convergence connects the finite-\(n\) Dicke code \(\ket{D_0^n},\ket{D_{\Nfp}^n}\) with the asymptotic protocol discussed below that lead to the golden ratio limit.

Fix an excitation cutoff \(K<\infty\) and let us restrict the input subspace \begin{equation} 
\cC_{n,K} = \operatorname{span} \{\ket{D_0^n},\ldots,\ket{D_K^n}\} \subseteq\cH_{n,\sym}=\operatorname{span} \{\ket{D_0^n},\ldots,\ket{D_n^n}\}. 
\end{equation} 
States in this subspace contain at most \(K\) excitations. After deletion, the number \(m\) of surviving particles may vary, but the number of surviving excitations cannot exceed \(\min(K,m)\). Therefore, the part of the \(m\)-survivor block that is reachable from \(\cC_{n,K}\) is 
\begin{equation} 
\cH_{m,K} = \operatorname{span} \{\ket{D_j^m}:0\le j\le\min(K,m)\}, \end{equation} 
and the complete reachable output space is 
\begin{equation}
    \cH_{n,K}^{\rm out} = \bigoplus_{m=0}^n\cH_{m,K}, 
    \text{ with }
    \rho=\bigoplus_{m=0}^np_m \rho_m,\quad
    p_m=\binom{n}{m}\eta^m(1-\eta)^{n-m}={\rm Pr}(J=m).
\end{equation}
To compare PI deletion with a pure-loss channel acting on a fixed Fock space, we subsequently retain \(j\) and remove the block label \(m\) through the packing map defined below. We place all survivor sectors in the common Fock space 
$\cH_K = \operatorname{span}\{\ket0,\ldots,\ket K\}$. 
For each \(m\), define the isometry 
\begin{equation} 
J_m:\cH_{m,K}\longrightarrow\cH_K, \qquad J_m\ket{D_j^m}=\ket j. \end{equation} 
Thus \(m\) records the total number of surviving particles, whereas \(j\) records the number of surviving excitations. The packing map retains \(j\) and discards the survivor-number label \(m\): \begin{equation} 
\cP_{n,K}(\omega) = \sum_{m=0}^n J_m\Pi_m\omega\Pi_mJ_m^\dagger, \label{eq:supp-fock-packing-map} 
\end{equation} 
where \(\Pi_m\) projects onto \(\cH_{m,K}\). The operators \(J_m\Pi_m\) are Kraus operators. Moreover, 
\begin{equation} 
\sum_{m=0}^n (J_m\Pi_m)^\dagger(J_m\Pi_m) = \sum_{m=0}^n\Pi_m = I_{\cH_{n,K}^{\rm out}}, 
\end{equation} 
so \(\cP_{n,K}\) is completely positive and trace preserving on every output that can arise from an input in \(\cC_{n,K}\). 
Let \(U_{n,K}:\cH_K\to\cC_{n,K}\) be the encoding isometry $U_{n,K}\ket k=\ket{D_k^n}$. 
Stochastic PI deletion followed by packing defines $\cT_{n,K,\eta}:\cH_K\mapsto \cH_K$ as
\begin{equation} 
\cT_{n,K,\eta}(\rho) = \cP_{n,K}\!\left[ \Del_{n,p}^{\rm PI} \bigl(U_{n,K}\rho U_{n,K}^\dagger\bigr) \right], \qquad p=1-\eta. \label{eq:supp-packed-deletion-channel} 
\end{equation} 
Both the input and output of \(\cT_{n,K,\eta}\) belong to the same finite-dimensional space \(\cH_K\), hence can be compare, in the same space, with the pure-loss channel that describes the large-\(n\) limit. For fixed \(K\), the packed PI-deletion channel approaches a process in which each of the finitely many excitations survives independently with probability \(\eta\). This is the finite-cutoff bosonic pure-loss channel 
\begin{equation} 
\cL_\eta(\rho) = \sum_{a=0}^{K}L_a\rho L_a^\dagger, \qquad L_a\ket{k} = \sqrt{\binom{k}{a}(1-\eta)^a\eta^{k-a}}\, \ket{k-a}. 
\label{eq:supp-pure-loss-channel} 
\end{equation} 
Here \(k\) is the initial excitation number and \(a\) is the number of excitations lost. Hence \(0\le a\le k\le K\), and \(L_a\ket{k}=0\) for \(a>k\). The corresponding probability of losing exactly \(a\) excitations is \(\binom{k}{a}(1-\eta)^a\eta^{k-a}\). To prove convergence, we compare the two channels on the matrix units \(\ket{k}\!\bra{\ell}\), which form a basis of \(\cL(\cH_K)\). For the pure-loss channel, 
\begin{equation} 
\cL_\eta(\ket k\!\bra\ell) = \sum_{a=0}^{\min(k,\ell)} \sqrt{\binom{k}{a}\binom{\ell}{a}} (1-\eta)^a\eta^{(k+\ell)/2-a} \ket{k-a}\!\bra{\ell-a}. 
\label{eq:supp-pure-loss-matrix-units} 
\end{equation}
\begin{remark}[Extension to the full survivor space]The packing map \(\cP_{n,K}\) is defined on the output space reached from inputs with at most \(K\) excitations, which is sufficient for the convergence theorem. For the sake of formality, to extend it to the full survivor-space direct sum, let \(\Pi_{\rm out}\) project onto \(\cH_{n,K}^{\rm out}\), choose any fixed state \(\tau_K\in\cL(\cH_K)\), and define
\begin{equation} 
\widehat\cP_{n,K}(\omega) = \cP_{n,K} \bigl(\Pi_{\rm out}\omega\Pi_{\rm out}\bigr) + \Tr\!\left[(I-\Pi_{\rm out})\omega\right]\tau_K. 
\label{eq:supp-full-packing-extension} 
\end{equation} 
This map first distinguishes the reachable subspace from its orthogonal complement. It applies the packing map to the former and replaces the latter by \(\tau_K\), and is therefore CPTP. Since deletion from \(\cC_{n,K}\) never leaves the reachable subspace, this extension agrees with \(\cP_{n,K}\), so it does not affect the convergence argument we are discussing below. 
\end{remark}
To compare the packed finite-\(n\) channel with pure loss, we evaluate both channels on the matrix units \(\ket{k}\!\bra{\ell}\), which form a basis of \(\cL(\cH_K)\). After encoding, \( \ket{k}\!\bra{\ell} \longmapsto \ket{D_k^n}\!\bra{D_\ell^n}. \) Let \(M_n\sim\Bin(n,\eta)\) be the random number of surviving particles. Conditioned on \(J_n=j\), the probability that \(b\) of the \(k\) input excitations remain among the survivors is 
\begin{equation} h_{n,j}(k,b) = \frac{\binom{k}{b}\binom{n-k}{j-b}}{\binom nj}. 
\label{eq:supp-hypergeom-survivor} 
\end{equation} 
If \(a\) excitations are lost, then \(b=k-a\). Hence \(h_{n,j}(k,k-a)\) is the conditional probability of losing exactly \(a\) excitations from \(\ket{D_k^n}\). Applying the exact-deletion formula of Lemma~\ref{lem:supp-exact-deletion-kraus}, followed by the packing identification \(J_j\ket{D_s^j}=\ket s\), gives 
\begin{equation} 
\cT_{n,K,\eta}(\ket k\!\bra\ell) = \sum_{a=0}^{\min(k,\ell)} \Gamma_{n,\eta}^{k\ell a} \ket{k-a}\!\bra{\ell-a}, \quad\Gamma_{n,\eta}^{k\ell a} = \E\!\left[ \sqrt{ h_{n,J_n}(k,k-a) h_{n,J_n}(\ell,\ell-a) } \right]. 
\label{eq:supp-Gamma-coeff} 
\end{equation} 
The same loss index \(a\) appears on the ket and bra because tracing out the deleted subsystem removes coherence between different lost-excitation numbers. The pure-loss channel has the same output matrix units, with coefficients 
\begin{equation} 
\Gamma_{\infty,\eta}^{k\ell a} = \sqrt{\binom{k}{a}\binom{\ell}{a}} (1-\eta)^a \eta^{(k+\ell)/2-a}. 
\label{eq:supp-Gamma-limit} 
\end{equation} 
The convergence of the two channels is therefore reduced to the convergence of the finitely many coefficients \(\Gamma_{n,\eta}^{k\ell a}\) to \(\Gamma_{\infty,\eta}^{k\ell a}\).



\begin{lemma}[Coefficient convergence]
\label{lem:supp-coefficient-convergence}
Fix \(K<\infty\) and a compact interval \(I=[\eta_-,\eta_+]\subset(0,1)\).  There is an explicit quantity \(\beta_{n,K,I}\to0\) such that, for all \(\eta\in I\) and all \(0\le a\le\min(k,\ell)\le K\),
\begin{equation}
\left|\Gamma_{n,\eta}^{k\ell a}-\Gamma_{\infty,\eta}^{k\ell a}\right|
\le \beta_{n,K,I}.
\label{eq:supp-coefficient-bound}
\end{equation}
One valid choice is as follows.  Let
$\delta_I=\frac12\min\{\eta_-,1-\eta_+\}$
and define
\begin{equation}
\alpha_{n,K,I}
=
\max_{\substack{0\le k\le K,\,0\le b\le k\\
\delta_I\le j/n\le1-\delta_I}}
\left|
 h_{n,j}(k,b)-\binom{k}{b}(j/n)^b(1-j/n)^{k-b}
\right|.
\end{equation}
Let
$G_{k\ell a}(x)=
\sqrt{\binom{k}{a}\binom{\ell}{a}}
(1-x)^a x^{(k+\ell)/2-a}$
and let \(L_{K,I}\) be the maximum Lipschitz constant of the finitely many functions \(G_{k\ell a}\) on \([\delta_I,1-\delta_I]\).  Then
\begin{equation}
\beta_{n,K,I}
=2\sqrt{\alpha_{n,K,I}}+\frac{L_{K,I}}{2\sqrt n}+8e^{-2n\delta_I^2}
\label{eq:supp-beta-bound}
\end{equation}
has the required property.
\end{lemma}

\begin{proof}
For \(j/n\in[\delta_I,1-\delta_I]\), Eq.~\eqref{eq:supp-hypergeom-survivor} is a finite product of factors of the form \((j-s)/n\), \((n-j-s)/n\), and \((n-s)/n\), with at most \(K\) factors.  Hence it converges uniformly to \(\binom{k}{b}(j/n)^b(1-j/n)^{k-b}\), and the finite maximum \(\alpha_{n,K,I}\) tends to zero.

For nonnegative numbers \(u,u',w,w'\in[0,1]\),
\begin{equation}
|\sqrt{uw}-\sqrt{u'w'}|
\le |\sqrt u-\sqrt{u'}|+|\sqrt w-\sqrt{w'}|
\le 2\sqrt{\max\{|u-u'|,|w-w'|\}}.
\end{equation}
Therefore, on the central event \(\Jsurv/n\in[\delta_I,1-\delta_I]\), the finite-\(n\) square-root coefficient differs from \(G_{k\ell a}(\Jsurv/n)\) by at most \(2\sqrt{\alpha_{n,K,I}}\).

Next compare \(G_{k\ell a}(\Jsurv/n)\) with \(G_{k\ell a}(\eta)\).  On the central interval, the Lipschitz bound gives
\begin{equation}
\E\left|G_{k\ell a}(\Jsurv/n)-G_{k\ell a}(\eta)\right|
\le L_{K,I}\E|\Jsurv/n-\eta|
\le \frac{L_{K,I}}{2\sqrt n},
\end{equation}
because \(\operatorname{Var}(\Jsurv/n)=\eta(1-\eta)/n\le1/(4n)\).  Finally, for \(\eta\in I\), the distance from \(\eta\) to either boundary of \([\delta_I,1-\delta_I]\) is at least \(\delta_I\).  Hoeffding's inequality gives
\begin{equation}
\Pr\left[\Jsurv/n\notin[\delta_I,1-\delta_I]\right]
\le 2e^{-2n\delta_I^2}.
\end{equation}
All coefficients are bounded in magnitude by one, so the discarded tail contributes at most \(8e^{-2n\delta_I^2}\) to the displayed bound.  Combining the three estimates proves Eq.~\eqref{eq:supp-beta-bound}.
\end{proof}

\begin{theorem}[Fixed-excitation PI deletion converges to pure loss]
\label{thm:supp-pi-to-pure-loss}
Fix \(K<\infty\) and a compact interval \(I=[\eta_-,\eta_+]\subset(0,1)\).  For all \(\eta\in I\),
\begin{equation}
\norm{\cT_{n,K,\eta}-\cL_\eta}_\diamond
\le
\epsilon_{n,K,I},
\qquad
\epsilon_{n,K,I}:=(K+1)^3\beta_{n,K,I},
\label{eq:supp-channel-convergence}
\end{equation}
where \(\beta_{n,K,I}\) is given in Eq.~\eqref{eq:supp-beta-bound}.  In particular, \(\epsilon_{n,K,I}\to0\).  Equivalently, for every finite reference system \(R\) and every state \(\rho_{RA}\) with \(A\simeq\cH_K\),
\begin{equation}
\left\|
(\operatorname{id}_R\otimes\cT_{n,K,\eta})(\rho_{RA})
-
(\operatorname{id}_R\otimes\cL_\eta)(\rho_{RA})
\right\|_1
\le \epsilon_{n,K,I}.
\end{equation}
For a fixed \(\eta\in(0,1)\), the same statement follows by choosing any compact interval \(I\subset(0,1)\) containing \(\eta\).
\end{theorem}

\begin{proof}
By Eqs.~\eqref{eq:supp-Gamma-coeff}--\eqref{eq:supp-Gamma-limit} and Lemma~\ref{lem:supp-coefficient-convergence}, for each matrix unit \(E_{k\ell}=\ket k\!\bra\ell\),
\begin{equation}
\norm{(\cT_{n,K,\eta}-\cL_\eta)(E_{k\ell})}_1
\le (K+1)\beta_{n,K,I}.
\end{equation}
Let \(d=K+1\).  For any operator \(X\) on \(R\otimes\cH_K\), write \(X=\sum_{k,\ell=0}^{K}X_{k\ell}\otimes E_{k\ell}\).  Each block satisfies \(\norm{X_{k\ell}}_1\le\norm{X}_1\).  Therefore
\begin{align}
\norm{(\operatorname{id}_R\otimes(\cT_{n,K,\eta}-\cL_\eta))(X)}_1
\le \sum_{k,\ell}\norm{X_{k\ell}}_1\,\norm{(\cT_{n,K,\eta}-\cL_\eta)(E_{k\ell})}_1
\le d^2(K+1)\beta_{n,K,I}\norm{X}_1.
\end{align}
By finite-dimensional stabilization, it is sufficient to take \(\dim R=K+1\).  The block estimate is independent of \(\dim R\), and therefore yields Eq.~\eqref{eq:supp-channel-convergence}.
\end{proof}

\begin{corollary}[Strict limiting violation transfers to finite \(n\)]
\label{cor:supp-strict-transfer}
Let a fixed finite-Fock code and fixed direct finite-Fock POVMs give CHSH value \(S_\infty=2+\Delta\) with \(\Delta>0\) under \(\cL_\eta\otimes\cL_\eta\).  Then the corresponding finite-\(n\) PI code obtained by \(\ket k\mapsto\ket{D_k^n}\), with the induced survivor-sector measurements obtained by the packing identification, gives \(S_n>2\) for all sufficiently large \(n\).
\end{corollary}

\begin{proof}
Let \(\rho_\infty\) and \(\rho_n\) denote the limiting and finite-\(n\) two-party output states after the same logical encoding.  If the one-sided channel distance is at most \(\epsilon_{n,K,I}\), then the two-sided output states differ by at most \(2\epsilon_{n,K,I}\) in trace norm, by a triangle inequality and contractivity of trace norm.  The CHSH operator built from Hermitian contractions obeys Tsirelson's norm bound \(\|\mathcal B_{\rm CHSH}\|_\infty\le2\sqrt2\), for instance by dilating contractions to dichotomic observables.  Therefore
\begin{equation}
|S_n-S_\infty|\le 2\sqrt2\,\norm{\rho_n-\rho_\infty}_1\le 4\sqrt2\,\epsilon_{n,K,I}.
\end{equation}
Since \(\epsilon_{n,K,I}\to0\), choose \(n\) large enough that \(4\sqrt2\,\epsilon_{n,K,I}<\Delta\).  Then \(S_n>2\).
\end{proof}


The finite-\(n\) PI implication follows only after invoking Corollary~\ref{cor:supp-strict-transfer}.

\begin{definition}[vacuum-anchored code and observables]
\label{def:supp-anchored-code}
For \(\Nfp\ge1\), define
$\ket{0_L}=\ket0$,
$\ket{1_L}=\ket{\Nfp}$.
On \(\operatorname{span}\{\ket0,\ldots,\ket{\Nfp}\}\), set
\begin{equation}
Z_{\Nfp}=\ket0\!\bra0-\sum_{j=1}^{\Nfp}\ket j\!\bra j,
\qquad
X_{\Nfp}=\ket0\!\bra{\Nfp}+\ket{\Nfp}\!\bra0.
\label{eq:supp-ZN-XN}
\end{equation}
\end{definition}

\begin{lemma}[Observable validity]
\label{lem:supp-anchored-observable-validity}
The observable \(Z_{\Nfp}\) is dichotomic.  The observable \(X_{\Nfp}\) is a Hermitian contraction.  Hence \(E_\pm^X=(I\pm X_{\Nfp})/2\) defines a valid binary POVM, with randomized output on the kernel of \(X_{\Nfp}\).  For every real \(\theta\),
\begin{equation}
B(\theta)=\cos\theta\,Z_{\Nfp}+\sin\theta\,X_{\Nfp}
\end{equation}
is also a Hermitian contraction.
\end{lemma}

\begin{proof}
The operator \(Z_{\Nfp}\) is diagonal with eigenvalues \(\pm1\).  The operator \(X_{\Nfp}\) is the Pauli \(X\) operator on \(\operatorname{span}\{\ket0,\ket{\Nfp}\}\) and is zero on the orthogonal subspace \(\operatorname{span}\{\ket1,\ldots,\ket{\Nfp-1}\}\), so its spectrum is contained in \(\{-1,0,1\}\).  Thus it is a Hermitian contraction.  On \(\operatorname{span}\{\ket0,\ket{\Nfp}\}\), the restrictions of \(Z_{\Nfp}\) and \(X_{\Nfp}\) anticommute, so \(B(\theta)\) has eigenvalues \(\pm1\) there.  On the kernel of \(X_{\Nfp}\), \(B(\theta)=-\cos\theta\,I\), whose norm is at most one.  Hence \(B(\theta)\) is a Hermitian contraction.
\end{proof}

\begin{lemma}[Anchored correlators]
\label{lem:supp-anchored-correlators}
Let
$\ket{\Phi_{\Nfp}}=\frac{\ket{00}+\ket{\Nfp\,\Nfp}}{\sqrt2}$
and apply \(\cL_\eta\otimes\cL_\eta\).
Then
\begin{equation}
\langle Z_{\Nfp}\otimes Z_{\Nfp}\rangle=c_{\Nfp}(\eta),
\quad
\langle X_{\Nfp}\otimes X_{\Nfp}\rangle=\eta^{\Nfp},
\text{  with }
c_{\Nfp}(\eta)=1-2q_{\Nfp}(1-q_{\Nfp}),
\quad
q_{\Nfp}=(1-\eta)^{\Nfp}.
\label{eq:supp-anchored-correlators}
\end{equation}
and the mixed correlators \(\langle Z_{\Nfp}\otimes X_{\Nfp}\rangle\) and \(\langle X_{\Nfp}\otimes Z_{\Nfp}\rangle\) vanish.
\end{lemma}

\begin{proof}
Under pure loss,
\begin{align}
\cL_\eta(\ket0\!\bra0)=\ket0\!\bra0,\qquad
\cL_\eta(\ket{\Nfp}\!\bra{\Nfp})=\sum_{j=0}^{\Nfp}\binom{\Nfp}{j}\eta^j(1-\eta)^{\Nfp-j}\ket j\!\bra j,\qquad
\cL_\eta(\ket0\!\bra{\Nfp})=\eta^{\Nfp/2}\ket0\!\bra{\Nfp}.
\end{align}
The diagonal component \(\ket{00}\!\bra{00}\) contributes \(+1\) to \(Z_{\Nfp}\otimes Z_{\Nfp}\).  In the component originating from \(\ket{\Nfp\,\Nfp}\), each side gives the local sign \(+1\) with probability \(q_{\Nfp}\) and \(-1\) with probability \(1-q_{\Nfp}\).  Its local sign expectation is \(2q_{\Nfp}-1\), so the two-party contribution is \((2q_{\Nfp}-1)^2\).  Averaging the two diagonal Bell components gives
$\frac12\left[1+(2q_{\Nfp}-1)^2\right]=1-2q_{\Nfp}(1-q_{\Nfp})$.
The only term contributing to \(X_{\Nfp}\otimes X_{\Nfp}\) is the two-party coherence \(
\ket{00}\!\bra{\Nfp\,\Nfp}+\ket{\Nfp\,\Nfp}\!\bra{00}
\), whose coefficient is multiplied by \(\eta^{\Nfp}\).  This gives the second equality.  The mixed correlators vanish because one observable is diagonal and the other couples only \(\ket0\) and \(\ket{\Nfp}\), while the state has no corresponding one-sided coherence.
\end{proof}

\begin{theorem}[Anchored pure-loss golden-ratio threshold]
\label{thm:supp-anchored-golden}
For the limiting pure-loss family in Definition~\ref{def:supp-anchored-code}, the CHSH value optimized in the \((Z_{\Nfp},X_{\Nfp})\) plane is
\begin{equation}
S_{\Nfp}(\eta)=2\sqrt{c_{\Nfp}(\eta)^2+\eta^{2\Nfp}}.
\label{eq:supp-SN}
\end{equation}
For \(\eta>1/2\), there exists a finite \(\Nfp\) with \(S_{\Nfp}(\eta)>2\) if and only if
\begin{equation}
\eta>\etaG=\frac{\sqrt5-1}{2}.
\end{equation}
For \(1/2\le\eta\le\etaG\), one has \(S_{\Nfp}(\eta)\le2\) for every finite \(\Nfp\).
\end{theorem}

\begin{proof}
Use Alice's observables \(A_0=Z_{\Nfp}\) and \(A_1=X_{\Nfp}\).  Bob uses
\begin{equation}
B_0=\cos\theta\,Z_{\Nfp}+\sin\theta\,X_{\Nfp},
\qquad
B_1=\cos\theta\,Z_{\Nfp}-\sin\theta\,X_{\Nfp}.
\end{equation}
By Lemma~\ref{lem:supp-anchored-observable-validity}, these are Hermitian contractions.  Lemma~\ref{lem:supp-anchored-correlators} gives
\begin{equation}
S(\theta)=2\left[c_{\Nfp}(\eta)\cos\theta+\eta^{\Nfp}\sin\theta\right].
\end{equation}
Optimizing over \(\theta\) yields Eq.~\eqref{eq:supp-SN}.
Let \(q=q_{\Nfp}=(1-\eta)^{\Nfp}\).  The condition \(S_{\Nfp}>2\) is equivalent to
\begin{align}
\eta^{2\Nfp}>1-c_{\Nfp}^2
=1-\left[1-2q(1-q)\right]^2
=4q(1-q)\left(1-q+q^2\right).
\label{eq:supp-golden-condition}
\end{align}

Assume first \(1/2\le\eta\le\etaG\).  Then \(\eta^2\le1-\eta\), so \(\eta^{2\Nfp}\le(1-\eta)^{\Nfp}=q\).  Also \(q\le1/2\), and for \(0\le q\le1/2\),
\begin{equation}
4(1-q)(1-q+q^2)>1.
\end{equation}
Indeed, \(f(q)=4(1-q)(1-q+q^2)\) is decreasing on \([0,1/2]\), since \(f'(q)=4(-2+4q-3q^2)<0\), and \(f(1/2)=3/2\).  Therefore the right-hand side of Eq.~\eqref{eq:supp-golden-condition} is strictly larger than \(q\), while the left-hand side is at most \(q\).  Hence no finite \(\Nfp\) violates CHSH.

Conversely, assume \(\eta>\etaG\).  Then
\begin{equation}
R:=\frac{\eta^2}{1-\eta}>1.
\end{equation}
Since \(\eta^{2\Nfp}=R^{\Nfp}q\) and
\begin{equation}
4q(1-q)(1-q+q^2)\le4q,
\end{equation}
any \(\Nfp\) with \(R^{\Nfp}>4\) satisfies Eq.~\eqref{eq:supp-golden-condition}.  Thus some finite \(\Nfp\) gives \(S_{\Nfp}>2\).  The boundary is the exponential balance \(\eta^2=1-\eta\), whose positive solution is \(\etaG\).
\end{proof}

\begin{corollary}[Finite-\(n\) PI implication]
\label{cor:supp-golden-finite-n}
For every \(\eta>\etaG\), choose \(\Nfp\) such that \(S_{\Nfp}(\eta)>2\) in Theorem~\ref{thm:supp-anchored-golden}.  Then the finite-\(n\) PI code \(\ket{D_0^n},\ket{D_{\Nfp}^n}\), with the survivor-sector measurements induced by the packing map of Sec.~\ref{supp:pure-loss-limit}, gives recovery-free CHSH violation for all sufficiently large \(n\).
\end{corollary}

\begin{proof}
This is exactly Corollary~\ref{cor:supp-strict-transfer} applied to the fixed finite-Fock code and measurements of Theorem~\ref{thm:supp-anchored-golden}.
\end{proof}
\subsection{Numerical view of the golden-ratio crossover}
\label{supp:golden-ratio-crossover}

The threshold in Theorem~\ref{thm:supp-anchored-golden} can be understood by separating the two terms that determine whether \(S_{\Nfp}\) exceeds the local bound. Define
\begin{equation}
\begin{aligned}
G_{\Nfp}(\eta)
:=\eta^{2\Nfp},\quad
D_{\Nfp}(\eta)
:=1-c_{\Nfp}(\eta)^2\
=4q_{\Nfp}(1-q_{\Nfp})
\bigl(1-q_{\Nfp}+q_{\Nfp}^2\bigr),
\end{aligned}
\qquad
q_{\Nfp}:=(1-\eta)^{\Nfp}.
\label{eq:supp-golden-gain-deficit}
\end{equation}
Here \(G_{\Nfp}\) is the coherent contribution to the squared CHSH value, whereas \(D_{\Nfp}\) is the deficit of the longitudinal correlation from its ideal value. Eq.~\eqref{eq:supp-SN} gives
$\frac{S_{\Nfp}(\eta)^2}{4}-1
=G_{\Nfp}(\eta)-D_{\Nfp}(\eta)$.
Therefore,
\begin{equation}
S_{\Nfp}(\eta)>2
\quad\Longleftrightarrow\quad
G_{\Nfp}(\eta)>D_{\Nfp}(\eta).
\label{eq:supp-golden-violation-condition}
\end{equation}
To display this competition over a wide range of \(\Nfp\), we introduce
$\Lambda_{\Nfp}(\eta)
:=
\log_{10}
\left[
\frac{G_{\Nfp}(\eta)}
{D_{\Nfp}(\eta)}
\right]$. 
Then \(\Lambda_{\Nfp}(\eta)>0 \Longleftrightarrow S_{\Nfp}(\eta)>2\). 
The logarithmic scale also resolves differences that become exponentially small at large \(\Nfp\).

\begin{figure}[t]
\centering
\includegraphics[width=0.95\linewidth]{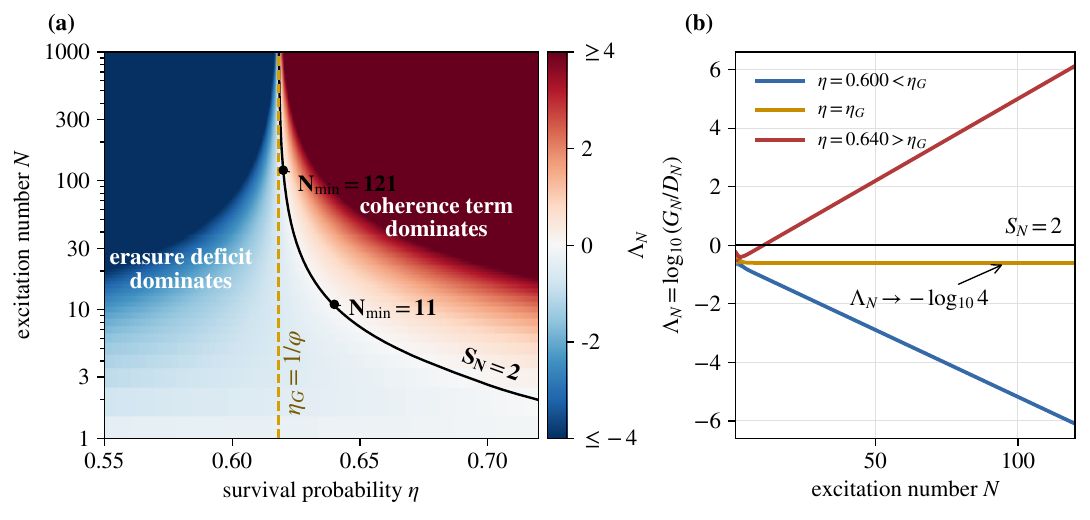}
\caption{Competition between coherence and complete-erasure errors in the vacuum-anchored \(\Nfp\)-excitation protocol. 
(a) The color scale shows the exact logarithmic ratio 
\(\Lambda_{\Nfp}(\eta)=\log_{10}[G_{\Nfp}(\eta)/D_{\Nfp}(\eta)]\). The black contour \(\Lambda_{\Nfp}=0\) is exactly the CHSH boundary \(S_{\Nfp}=2\); 
the region \(\Lambda_{\Nfp}>0\) violates CHSH. The vertical reference line marks \(\etaG=(\sqrt{5}-1)/2\).
(b) Cross sections at representative survival probabilities show the change in the large-\(\Nfp\) slope of \(\Lambda_{\Nfp}\). 
The slope is negative below \(\etaG\), vanishes at \(\etaG\), and is positive above \(\etaG\).}
\label{fig:supp-golden-ratio-balance}
\end{figure}

The two quantities in Eq.~\eqref{eq:supp-golden-gain-deficit} describe opposite effects of increasing \(\Nfp\). The probability
$q_{\Nfp}=(1-\eta)^{\Nfp}$ 
is the probability that all \(\Nfp\) excitations are lost on one side. 
This event maps the \(\ket{\Nfp}\) component to the vacuum and makes it indistinguishable from the logical state \(\ket{0}\). 
If complete erasure occurs on only one side, the longitudinal outcomes can disagree. This reduces \(c_{\Nfp}\) and produces the deficit \(D_{\Nfp}\).
Increasing \(\Nfp\) suppresses complete erasure exponentially. It is sufficient that one excitation survive to distinguish the \(\ket{\Nfp}\) branch from the vacuum. 
Consequently, for $\Nfp\to\infty$,
\begin{equation}
q_{\Nfp}\longrightarrow 0,
\qquad
c_{\Nfp}\longrightarrow 1,
\qquad
D_{\Nfp}\longrightarrow 0.
\end{equation}
 In this sense, the excitation number provides redundancy for the population information measured by \(Z_{\Nfp}\).
The same increase in \(\Nfp\) has the opposite effect on coherence. Under pure loss,
$\cL_{\eta}\left(\ket{0}\bra{\Nfp}\right)
=
\eta^{\Nfp/2}\ket{0}\bra{\Nfp}$.
Both parties must preserve this coherence for the \(X_{\Nfp}\otimes X_{\Nfp}\) correlation to survive. Hence
$\left\langle
X_{\Nfp}\otimes X_{\Nfp}
\right\rangle=
\eta^{\Nfp}
$, 
and its contribution to the optimized squared CHSH value is
$G_{\Nfp}=\eta^{2\Nfp}$.
This term also decreases exponentially with \(\Nfp\). 
A larger excitation number therefore suppresses complete-erasure errors but makes the vacuum--\(\Nfp\) coherence more fragile.
The relative decay rates follow from the exact ratio
\begin{equation}
\frac{G_{\Nfp}(\eta)}{D_{\Nfp}(\eta)}=
\frac{1}{
4(1-q_{\Nfp})
\bigl(1-q_{\Nfp}+q_{\Nfp}^2\bigr)
}
\left(
\frac{\eta^2}{1-\eta}
\right)^{\Nfp}.
\label{eq:supp-golden-exact-ratio}
\end{equation}
Since \(q_{\Nfp}\to 0\), its large-\(\Nfp\) form is
\begin{equation}
\Lambda_{\Nfp}(\eta)=
\Nfp\log_{10}
\left(
\frac{\eta^2}{1-\eta}
\right)
-\log_{10}4
+o(1).
\label{eq:supp-golden-ratio-asymptotic}
\end{equation}
This expression explains the two regions in Fig.~\ref{fig:supp-golden-ratio-balance}.
For
$\frac{1}{2}\leq\eta<\etaG$,
one has
$\frac{\eta^2}{1-\eta}<1$.
The coherent term then decreases faster with \(\Nfp\) than the complete-erasure deficit. Increasing \(\Nfp\) improves the longitudinal correlation, but the available coherence is lost even more rapidly. 
Thus \(D_{\Nfp}\) remains larger than \(G_{\Nfp}\), and the protocol remains in the nonviolating region of panel~(a).
For
$\eta>\etaG$,
the ordering of the decay rates is reversed:
$\frac{\eta^2}{1-\eta}>1$.
Complete erasure is then suppressed faster than the squared coherent contribution. For small \(\Nfp\), the finite prefactor in \(D_{\Nfp}\) may still prevent a violation. 
As \(\Nfp\) increases, however, \(G_{\Nfp}/D_{\Nfp}\) grows and eventually crosses one. 
This crossing is the black contour in panel~(a). It marks the point at which the remaining coherence becomes sufficient to compensate for the longitudinal deficit.
The boundary between these two behaviors is obtained when the exponential bases are equal:
$\eta^2=1-\eta$.
The positive solution is
$\eta=\etaG
=\frac{\sqrt{5}-1}{2}
=
\frac{1}{\varphi}$, where
$\varphi=\frac{1+\sqrt{5}}{2}$.
At this value,
$q_{\Nfp}
=(1-\etaG)^{\Nfp}
 \etaG^{2\Nfp}
G_{\Nfp}$.
The complete-erasure probability and the coherent CHSH contribution therefore have the same exponential dependence on \(\Nfp\).

Equality of the decay rates does not imply equality of the full finite-\(\Nfp\) contributions. At \(\eta=\etaG\), Eq.~\eqref{eq:supp-golden-exact-ratio} becomes
$\frac{G_{\Nfp}(\etaG)}{D_{\Nfp}(\etaG)}=
\frac{1}{
4(1-q_{\Nfp})
\bigl(1-q_{\Nfp}+q_{\Nfp}^2\bigr)
}$.
Consequently,
$\lim_{\Nfp\to\infty}
\frac{G_{\Nfp}(\etaG)}{D_{\Nfp}(\etaG)}=
\frac{1}{4},
\qquad
\lim_{\Nfp\to\infty}
\Lambda_{\Nfp}(\etaG)=
-\log_{10}4$.
The longitudinal deficit remains larger because of its finite prefactor. The balance at \(\etaG\) therefore does not itself produce a CHSH violation. 
The inverse golden ratio identifies the crossover between the two exponential rates, rather than an equality of the complete finite-\(\Nfp\) terms.

The black contour approaches \(\etaG\) only as \(\Nfp\to\infty\). Close to the threshold from above, Eq.~\eqref{eq:supp-golden-ratio-asymptotic} gives
$\Nfp_{\min}
\simeq
\frac{\ln 4}{
\ln!\left[\eta^2/(1-\eta)\right]
}$.
Expanding the denominator around \(\eta=\etaG\) yields
$\Nfp_{\min}
\sim
\frac{0.2368}{\eta-\etaG}$
as 
$\eta\downarrow\etaG$.
Thus the excitation number required to cross the CHSH boundary grows rapidly near the golden-ratio threshold.
For every fixed \(\eta<1\), both \(G_{\Nfp}\) and \(D_{\Nfp}\) vanish as \(\Nfp\to\infty\), and therefore
$S_{\Nfp}(\eta)\longrightarrow 2$.
The logarithmic ratio in Fig.~\ref{fig:supp-golden-ratio-balance} determines which contribution is dominant, but it does not measure the absolute magnitude of the violation. 
Very close to \(\etaG\), a violation exists for a sufficiently large finite \(\Nfp\), but its magnitude is correspondingly small.


\subsection{Measurement robustness $v$-- One-excitation protocol}
\label{supp:visibility}
The parameter \(v\) is a local measurement visibility.  It attenuates each local transverse signed observable.  
Assume sharp logical \(Z\) measurement and local transverse visibility \(v\).  In the \(ZX\) plane, the signed observable is
$O_v(\theta)=\cos\theta\,Z+v\sin\theta\,X$.
For \(v<1\), this is generally an unsharp signed observable: its norm is at most one, and \(E_\pm=(I\pm O_v(\theta))/2\) is a valid binary POVM.

\begin{theorem}[One-excitation optimized visibility value]
\label{thm:supp-one-ex-visibility}
For the one-excitation amplitude-damping state with symmetric survival \(\eta=1-p\), sharp logical \(Z\), and local transverse visibility \(v\), the effective correlation singular values are $a=v^2\eta$ and $c=1-2\eta(1-\eta)$.
The optimized trusted-measurement CHSH value is
\begin{equation}
S_{\max}(v,\eta)=2\sqrt{v^4\eta^2+
\max\left[v^4\eta^2,(1-2\eta(1-\eta))^2\right]}.
\label{eq:supp-Smax-vp}
\end{equation}
\end{theorem}

\begin{proof}
For the ideal one-excitation output, the three nonzero diagonal entries of the two-qubit correlation matrix can be chosen as \(\eta\), \(-\eta\), and \(c=1-2\eta(1-\eta)\).  Local transverse visibility sends \(X\mapsto vX\) and \(Y\mapsto vY\), so the two transverse correlation magnitudes become \(a=v^2\eta\), while the longitudinal value remains \(c\).  The Horodecki criterion says that the optimized CHSH value for a two-qubit correlation matrix is \(2\sqrt{u_1+u_2}\), where \(u_1,u_2\) are the two largest eigenvalues of \(T^T T\) \cite{Horodecki1995}.  Here the squared singular values are \(a^2,a^2,c^2\), which gives Eq.~\eqref{eq:supp-Smax-vp}.
\end{proof}

There are two branches.  If \(a\ge c\), the two largest singular values are transverse and \(S_{\max}=2\sqrt2\,a\).  If \(c\ge a\), the optimal directions are one longitudinal and one transverse, and \(S_{\max}=2\sqrt{c^2+a^2}\).

\begin{proposition}[Leading estimate]
\label{prop:supp-pc-small-v}
In the mixed longitudinal-transverse branch, the threshold \(\eta_c(v)\) satisfies
\begin{equation}
(1-2\eta_c(1-\eta_c))^2+v^4\eta_c^2=1.
\label{eq:supp-mixed-threshold-equation}
\end{equation}
For small thresholds in the experimentally relevant range discussed in the Letter, the leading estimate is
$1-\eta_c\simeq\frac{v^4}{4}$.
\end{proposition}

\begin{proof}
In the mixed branch, the threshold is \(2\sqrt{c^2+a^2}=2\), i.e. \(c^2+a^2=1\).  Substituting \(c=\) and \(a\) gives Eq.~\eqref{eq:supp-mixed-threshold-equation}.  Expanding at small \(1-\eta_c\) and imposing to be equal to 1,
gives the leading estimate \(1-\eta_c\simeq v^4/4\).  
The exact curve in the figures is obtained by solving Eq.~\eqref{eq:supp-mixed-threshold-equation}.
\end{proof}

\subsection{Measurement robustness $v$ -- $N$--excitation protocol}
We now apply the same visibility model to the \(N\)-excitation family.  If the local measurement of the coherence between \(\ket{0}\) and \(\ket{\Nfp}\) has visibility \(v\), then \(X_{\Nfp}\mapsto vX_{\Nfp}\) on each side.  The transverse correlation therefore becomes
$\langle X_{\Nfp}\otimes X_{\Nfp}\rangle
=v^2\eta^{\Nfp}$.
and the tilted-measurement optimization gives
$S_{\Nfp,v}(\eta)=2\sqrt{c_{\Nfp}(\eta)^2+v^4\eta^{2\Nfp}}$.
For any fixed \(v>0\) independent of \(\Nfp\), the exponent ratio in the proof of Theorem~\ref{thm:supp-anchored-golden} becomes
$v^4\left(\frac{\eta^2}{1-\eta}\right)^{\Nfp}$.
The visibility changes the excitation number required to obtain a violation, but not the asymptotic threshold: a sufficiently large finite \(\Nfp\) exists whenever
\(\eta^2/(1-\eta)>1\), i.e. \(\eta>\etaG\).
The conclusion can change if the visibility decreases with \(\Nfp\).  Writing it as \(v_{\Nfp}\), define
$\chi_+:=\limsup_{\Nfp\to\infty} v_{\Nfp}^{4/\Nfp}$.
%
A sufficient condition for violation is
$\frac{\eta^2}{1-\eta}\chi_+>1$.
Indeed, this condition gives a subsequence along which \(v_{\Nfp}^4[\eta^2/(1-\eta)]^{\Nfp}\) grows exponentially. This is only a sufficient asymptotic condition; it is neither a complete characterization of \(\Nfp\)-dependent visibility nor a statement about any particular experimental platform.

\begin{figure}
\centering
\includegraphics[scale=0.9]{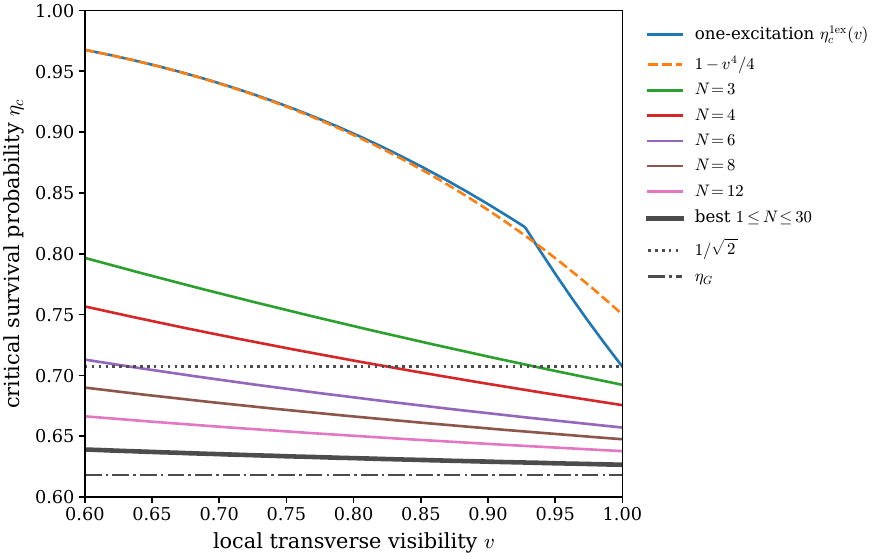}
\caption{
Critical survival probability \(\eta_c\) versus local measurement visibility \(v\). CHSH violation occurs for \(\eta>\eta_c\), so lower curves correspond to greater loss tolerance. The blue one-excitation curve changes slope because the CHSH optimum switches from the mixed longitudinal--transverse branch to the transverse--transverse branch; the dashed line \(\eta_c\simeq1-v^4/4\) is the small-loss approximation to the former. Curves labelled by \(N\) refer to the vacuum--\(N\)-excitation encodings, while the thick gray curve is the best threshold obtained over \(1\le N\le30\) at each \(v\).}
\label{fig:supp-combined-robustness}
\end{figure}
\subsection{Extended numerical search}
\label{supp:numerics}

The golden-ratio threshold of Theorem~\ref{thm:supp-anchored-golden} is exact for the anchored family \((\ket{00}+\ket{\Nfp\Nfp})/\sqrt2\), but the theorem does not establish optimality over all finite-Fock encodings and measurements.  The purpose of this section is therefore to test whether the threshold can be lowered by allowing progressively more general logical codes and observables.  In particular, we ask whether a finite-Fock strategy can produce a CHSH violation for \(\eta\leq\etaG\) by moving the encoded coherence away from the vacuum, combining several Fock-level coherences, or removing the measurement restrictions altogether.  The search provides numerical evidence within finite cutoffs; it is not an optimality proof.  All scripts, parameters, code coefficients, observables, and numerical outputs are included in the accompanying reproducibility archive.


We truncate the local Fock space at excitation number \(K\),
\begin{equation}
\cH_K=\operatorname{span}\{\ket0,\ket1,\ldots,\ket K\},
\end{equation}
and implement the pure-loss channel using the Kraus operators
\begin{equation}
L_a\ket{k}
=
\sqrt{
\binom{k}{a}
(1-\eta)^a
\eta^{k-a}
}\,
\ket{k-a},
\qquad
0\leq a\leq k\leq K.
\label{eq:supp-numerical-loss-kraus}
\end{equation}
Here \(a\) is the number of lost excitations and \(\eta\) is the survival probability of each excitation.

A code is specified by an isometry
\begin{equation}
V:\mathbb C^2\longrightarrow\cH_K,
\qquad
V\ket0=\ket{0_L},
\qquad
V\ket1=\ket{1_L},
\end{equation}
whose two columns are orthonormal logical states.  The corresponding encoded Bell state and its two-sided lossy output are
\begin{equation}
\ket{\Phi_L(V)}
=
\frac{
\ket{0_L0_L}
+
\ket{1_L1_L}
}{\sqrt2},
\qquad
\rho_\eta(V)
=
(\cL_\eta\otimes\cL_\eta)
\left(
\ket{\Phi_L(V)}\!\bra{\Phi_L(V)}
\right).
\label{eq:supp-numerical-objective}
\end{equation}

For each output state, we evaluate
\begin{equation}
S
=
\Tr\!\left[
\rho_\eta(V)
\left(
A_0\otimes(B_0+B_1)
+
A_1\otimes(B_0-B_1)
\right)
\right],
\label{eq:supp-numerical-chsh}
\end{equation}
where \(A_0,A_1,B_0,B_1\) are Hermitian contractions.  Each contraction \(O\), satisfying \(\norm{O}_\infty\leq1\), defines a valid binary POVM through \(E_\pm=(I\pm O)/2\).  We compare four search classes, ordered from the anchored analytical construction to increasingly general finite-Fock measurements.

\paragraph{Anchored Fock-pair benchmark.}

The anchored family uses
\begin{equation}
\ket{0_L}=\ket0,
\qquad
\ket{1_L}=\ket{\Nfp},
\qquad
1\leq\Nfp\leq K.
\end{equation}
It is called ``anchored'' because one logical state is fixed at the vacuum, and ``Fock-pair'' because the code is supported on the two Fock levels \(0\) and \(\Nfp\).  Its transverse observable probes the single coherence between these levels,
\begin{equation}
X_{\Nfp}
=
\ket0\!\bra{\Nfp}
+
\ket{\Nfp}\!\bra0,
\end{equation}
while the longitudinal observable distinguishes the vacuum from the nonvacuum sector.  For each cutoff \(K\), the benchmark is the largest closed-form value
\begin{equation}
S_{\Nfp}(\eta)
=
2\sqrt{
c_{\Nfp}(\eta)^2
+
\eta^{2\Nfp}
}
\end{equation}
over \(1\leq\Nfp\leq K\).  This class reproduces the analytical family of Theorem~\ref{thm:supp-anchored-golden}; no numerical optimization of its observables is required.

\paragraph{Single-coherence search.}

The single-coherence class allows the logical states \(\ket{0_L}\) and \(\ket{1_L}\) to be arbitrary orthonormal vectors in \(\cH_K\), but restricts the measurement to one selected Fock-level coherence.  For a pair \(0\leq r<s\leq K\), the transverse observable is
\begin{equation}
X_{rs}
=
\ket r\!\bra s
+
\ket s\!\bra r.
\end{equation}
The accompanying longitudinal observable is diagonal,
\begin{equation}
Z=\sum_{k=0}^{K}z_k\ket k\!\bra k,
\qquad
z_k\in\{+1,-1\},
\end{equation}
with \(z_r=+1\), \(z_s=-1\), and the remaining signs optimized by enumeration.  The local CHSH observables are then optimized within the two-dimensional operator space spanned by \(Z\) and \(X_{rs}\).

Thus, ``single coherence'' describes the measurement restriction, not the support of the code: the encoded states may occupy several Fock levels, but the transverse measurement accesses only one off-diagonal matrix element \(r\leftrightarrow s\) at a time.  The search compares all pairs \((r,s)\), all allowed diagonal sign patterns, the deterministic code list, and the sampled random isometries.  For each fixed code, this tests whether shifting the measured coherence away from \(0\leftrightarrow\Nfp\), changing the longitudinal assignment, or spreading the logical states across several Fock levels improves on the anchored family.

\paragraph{Band-limited search.}

The band-limited class permits each observable to contain several Fock-level coherences, while restricting how far apart the coupled levels may be.  A Hermitian contraction \(O\) has bandwidth \(B\) if
\begin{equation}
\bra r O\ket s=0
\qquad
\text{whenever}
\qquad
|r-s|>B.
\label{eq:supp-bandwidth-definition}
\end{equation}
For example, \(B=1\) permits diagonal terms and nearest-neighbor coherences, while larger \(B\) allows couplings between more widely separated Fock levels.  The four CHSH observables are otherwise independent.

The values reported as ``band-limited'' are the largest values found over the tested cutoffs, bandwidths, codes, random initializations, and alternating optimization runs.  They are feasible CHSH values, but they are not certified maxima of the band-limited class.  In particular, ``best band-limited'' always means ``best band-limited value found in the specified search.''

\paragraph{General-observable finite-Fock benchmark.}

In the final class, the four observables are arbitrary Hermitian contractions on \(\cH_K\), without diagonal, single-coherence, or bandwidth restrictions.  This is the most permissive measurement class included in the search.  We refer to it in the tables as the ``unrestricted benchmark,'' but the term ``unrestricted'' applies only to the observables within the chosen finite-dimensional space.  The calculation still uses finite cutoffs, a finite list of deterministic and random codes, and a nonconvex alternating optimization.  It is therefore not an optimization over all finite-Fock codes, nor over the full infinite-dimensional output space.  Its role is to test whether removing the imposed measurement structure reveals a violation missed by the more operationally constrained classes.

\paragraph{Alternating optimization.}

For the band-limited and general-observable classes, the observables are optimized by a see-saw procedure.  The idea is to optimize one observable at a time while keeping the state and the other three observables fixed.

For example, with \(B_0\) and \(B_1\) fixed, the dependence of the CHSH value on \(A_0\) is
\begin{equation}
\Tr(A_0H_{A_0}),
\qquad
H_{A_0}
=
\Tr_B\!\left[
\rho_\eta(V)
\bigl(I\otimes(B_0+B_1)\bigr)
\right].
\label{eq:supp-A0-effective-operator}
\end{equation}
If
\begin{equation}
H_{A_0}
=
\sum_j\lambda_j\ket{j}\!\bra{j}
\end{equation}
is its spectral decomposition, define
\begin{equation}
\operatorname{sign}(H_{A_0})
=
\sum_j\operatorname{sign}(\lambda_j)
\ket{j}\!\bra{j}.
\end{equation}
Among all Hermitian contractions \(A_0\), the exact maximizer of \(\Tr(A_0H_{A_0})\) is
\begin{equation}
A_0
\leftarrow
\operatorname{sign}\!\left(
\Tr_B\!\left[
\rho_\eta(V)
\bigl(I\otimes(B_0+B_1)\bigr)
\right]
\right).
\label{eq:supp-seesaw-update}
\end{equation}
Operationally, this update assigns the outcome \(+1\) to the nonnegative eigenspace of \(H_{A_0}\) and \(-1\) to its negative eigenspace.  The remaining updates are obtained in the same way:
\begin{align}
A_1
&\leftarrow
\operatorname{sign}\!\left(
\Tr_B\!\left[
\rho_\eta(V)
\bigl(I\otimes(B_0-B_1)\bigr)
\right]
\right),\\
B_0
&\leftarrow
\operatorname{sign}\!\left(
\Tr_A\!\left[
\rho_\eta(V)
\bigl((A_0+A_1)\otimes I\bigr)
\right]
\right),\\
B_1
&\leftarrow
\operatorname{sign}\!\left(
\Tr_A\!\left[
\rho_\eta(V)
\bigl((A_0-A_1)\otimes I\bigr)
\right]
\right).
\end{align}

Each update is exact for the observable being optimized while the state and the other observables are fixed.  This does not make the complete see-saw globally optimal: the CHSH objective is not jointly convex in all four observables, and different initial conditions may converge to different stationary points.  Moreover, the code is selected from deterministic candidates and randomly sampled isometries rather than optimized globally.  Multiple random starts and code samples are therefore used to reduce, but not eliminate, the risk of missing a better strategy.

The unrestricted sign update generally violates the bandwidth constraint in Eq.~\eqref{eq:supp-bandwidth-definition}.  For the band-limited search, two admissible candidates are constructed at each step.  First, the unrestricted operator \(\operatorname{sign}(H)\) is projected onto the band by deleting all matrix elements with \(|r-s|>B\), and the result is rescaled if necessary to restore operator norm at most one.  Second, \(H\) itself is projected onto the same band and normalized to a Hermitian contraction.  The candidate giving the larger one-step objective is retained.  These are the ``projected sign'' and ``projected linear'' updates, respectively.

This projected procedure is heuristic.  Projection of the unrestricted optimizer need not produce the optimal band-limited contraction.  The exact one-observable update would instead require the semidefinite program
\begin{equation}
\begin{aligned}
\text{maximize}\quad
&\Tr(OH),\\
\text{subject to}\quad
&O=O^\dagger,\qquad -I\leq O\leq I,\\
&\bra rO\ket s=0
\quad\text{for}\quad |r-s|>B.
\end{aligned}
\label{eq:supp-band-limited-sdp}
\end{equation}
This SDP was not solved in the present search.  The reported band-limited values are therefore achieved lower bounds on the optimum within the corresponding search class, not SDP-certified optima.

\paragraph{Search configuration and validation}

Table~\ref{tab:supp-numerical-grid} gives the complete search configuration.  The random-number-generator seed fixes all random code samples and initial observables, making the calculation reproducible.  The profile name \texttt{serious} identifies the stated numbers of samples, starts, and iterations.  A random code is obtained by sampling two complex vectors in \(\cH_K\) and orthonormalizing them to form an isometry \(V\).

The deterministic code list contains every anchored pair \(\ket0,\ket s\), with \(1\leq s\leq K\), and, for \(K\geq4\), the loss-code-inspired candidate denoted by \texttt{binomial\_024},
\begin{equation}
\ket{0_L}
=
\frac{\ket0+\ket4}{\sqrt2},
\qquad
\ket{1_L}
=
\ket2.
\end{equation}
The violation tolerance \(10^{-7}\) means that a value is reported as a violation only when \(S>2+10^{-7}\).  Smaller excesses are treated as numerical roundoff.  The convergence tolerance controls termination of the alternating updates, while the number of starts specifies how many independently initialized see-saw runs are performed for each code.  The bandwidth condition \(B\leq K\) prevents requesting matrix elements outside the truncated Fock space.

\begin{table}[t]
\centering
\caption{Configuration of the extended numerical search.}
\label{tab:supp-numerical-grid}
\begin{tabular}{ll}
\toprule
Quantity & Value \\
\midrule
Survival grid & \(\{1/\sqrt2,\etaG,0.60,0.58,0.56,0.54,0.52,0.505\}\) \\
Violation tolerance & \(10^{-7}\) \\
Convergence tolerance in code & \(10^{-12}\) \\
Anchored benchmark cutoffs & \(K\in\{2,3,4,5,6,8\}\) \\
Single-coherence cutoffs & \(K=2,3,4,5,6\) \\
Single-coherence code list & deterministic + 256 random isometries \\
Band-limited cutoffs & \(K\in\{3,4,5,6,8\}\) \\
Bandwidths & \(B=1,2,3,4\), with \(B\le K\) \\
Band-limited code list & deterministic + 64 random isometries \\
Band-limited see-saw & 12 starts/code, 150 iterations \\
Unrestricted cutoffs & \(K\in\{3,4,5,6,8\}\) \\
Unrestricted code list & deterministic + 128 random isometries \\
Unrestricted see-saw & 24 starts/code, 250 iterations \\
Deterministic candidates & anchored \(\ket0,\ket s\), \(1\le s\le K\); plus \texttt{binomial\_024} for \(K\ge4\) \\
\bottomrule
\end{tabular}
\end{table}

Before running the broader search, we checked the implementation against the analytical anchored formula for \(\Nfp=1,2,3,4\).  The largest absolute difference between the direct Kraus-channel calculation and the closed-form CHSH value was
$6.661338147750939\times10^{-16}.$
This agreement, at the level of floating-point precision, validates the channel implementation and the CHSH evaluation for the anchored test cases.

\paragraph{Numerical results}

Table~\ref{tab:supp-best-values} reports the largest value found in each search class at every tested survival probability.  At \(\eta=1/\sqrt2\), the anchored and single-coherence classes both reach \(S=2.033217\).  Allowing several coherences improves the best value to \(S=2.172382\) in both the band-limited and unrestricted searches.  Thus, broader measurements can substantially increase the magnitude of the violation when \(\eta\) is sufficiently above \(\etaG\).

The behavior changes at the golden-ratio value.  At \(\eta=\etaG\), the largest anchored value within the tested cutoff is \(1.998642\), consistent with the analytical result that no finite \(\Nfp\) violates CHSH at the boundary, although the value approaches \(2\) as \(\Nfp\) increases.  The three broader searches return \(2.000000\), but these values correspond to the local bound within numerical precision.  They are not violations.  The same pattern persists at every tested value below \(\etaG\): the anchored observable remains below \(2\), while the broader classes can recover the trivial local-bound value but do not exceed it by the reporting tolerance.

\begin{table}[t]
\centering
\caption{Best value by survival probability and regime.  Values at or below \(\etaG\) that round to \(2.000000\) are local-bound values within numerical precision, not robust violations above the tolerance.}
\label{tab:supp-best-values}
\begin{tabular}{ccccc}
\toprule
\(\eta\) & Anchored & Single coherence & Band-limited & Unrestricted benchmark \\
\midrule
0.707107 & 2.033217 & 2.033217 & 2.172382 & 2.172382 \\
0.618034 & 1.998642 & 2.000000 & 2.000000 & 2.000000 \\
0.600000 & 1.997663 & 2.000000 & 2.000000 & 2.000000 \\
0.580000 & 1.996295 & 2.000000 & 2.000000 & 2.000000 \\
0.560000 & 1.994482 & 2.000000 & 2.000000 & 2.000000 \\
0.540000 & 1.992050 & 2.000000 & 2.000000 & 2.000000 \\
0.520000 & 1.988789 & 2.000000 & 2.000000 & 2.000000 \\
0.505000 & 1.985652 & 2.000000 & 2.000000 & 2.000000 \\
\bottomrule
\end{tabular}
\end{table}

Table~\ref{tab:supp-best-below-etaG} resolves the apparent \(S=2\) values more closely.  In the single-coherence, band-limited, and unrestricted searches, the largest positive excesses below or at \(\etaG\) are between \(10^{-15}\) and \(10^{-14}\).  These numbers are comparable to machine-precision rounding errors and are eight orders of magnitude below the adopted violation tolerance \(10^{-7}\).  They therefore carry no physical significance.  The anchored entry remains below the local bound by \(1.358\times10^{-3}\) at the largest tested cutoff.

\begin{table}[t]
\centering
\small
\caption{Best entries at or below \(\etaG\) by regime.  The tiny positive excesses shown are numerical noise below the \(10^{-7}\) reporting tolerance, not physical violations.}
\label{tab:supp-best-below-etaG}
\resizebox{.65\linewidth}{!}{%
\begin{tabular}{@{}lllll@{}}
\toprule
Regime & Best \(\eta\) & Parameters & Code & \(S-2\) \\
\midrule
Anchored & 0.618033988750 & \(K=8\) & \texttt{anchored\_0\_8} & \(-1.358\times10^{-3}\) \\
Single-coh. & 0.580000000000 & \(K=4\) & \texttt{anchored\_0\_3} & \(+1.3\times10^{-15}\) \\
Band-limited & 0.505000000000 & \(K=8,B=4\) & \texttt{random\_50} & \(+1.8\times10^{-15}\) \\
Unrestricted & 0.600000000000 & \(K=8\) & \texttt{random\_57} & \(+6.2\times10^{-15}\) \\
\bottomrule
\end{tabular}%
}
\end{table}

Figure~\ref{fig:supp-best-S-vs-eta} summarizes the same comparison across the survival grid.  The vertical lines distinguish the exact anchored threshold \(\etaG\) from the one-excitation threshold \(1/\sqrt2\).  The improvement at \(1/\sqrt2\) shows that the broader search is capable of finding strategies beyond the anchored single-coherence benchmark.  Its failure to find a corresponding improvement below \(\etaG\) is therefore not simply due to all four search classes being numerically identical.

\begin{figure}[t]
\centering
\includegraphics[width=.78\linewidth]{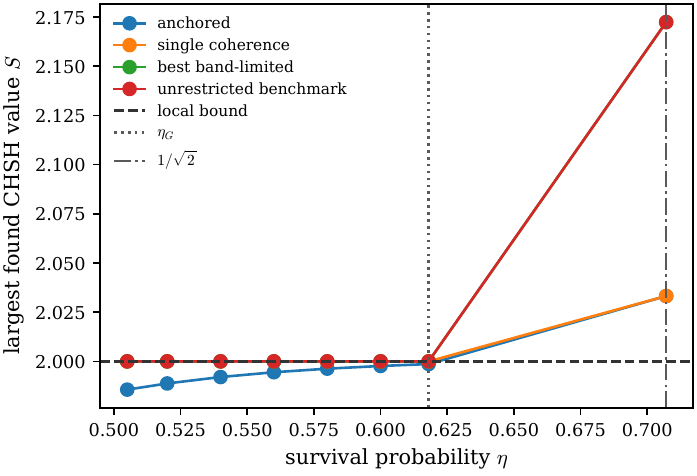}
\caption{Largest CHSH value found at each tested survival probability.  The anchored curve uses the closed-form \(N\)-excitation expression, the single-coherence search permits one optimized Fock-level coherence, the band-limited search permits several coherences with \(|r-s|\leq B\), and the unrestricted benchmark permits arbitrary Hermitian contractions within the tested cutoff.  Vertical lines mark \(\etaG\) and \(1/\sqrt2\).}
\label{fig:supp-best-S-vs-eta}
\end{figure}

Figure~\ref{fig:supp-performance-bandwidth} isolates the dependence on the bandwidth \(B\).  Increasing \(B\) enlarges the admissible measurement class by allowing coherences between more widely separated Fock levels.  The plotted values are the best feasible values found by the projected see-saw procedure at each bandwidth.  Because the update is not the exact SDP in Eq.~\eqref{eq:supp-band-limited-sdp}, the curves describe the performance of the implemented search and should not be interpreted as the optimal CHSH value at fixed \(B\).

\begin{figure}[t]
\centering
\includegraphics[width=.78\linewidth]{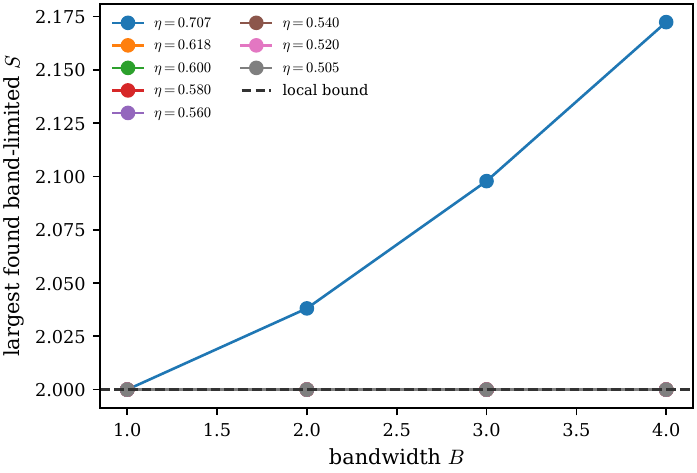}
\caption{Largest CHSH values found by the band-limited search as a function of the bandwidth \(B\).  A bandwidth \(B\) permits matrix elements between Fock levels satisfying \(|r-s|\leq B\).  The values are obtained with projected sign and linear updates and are feasible lower bounds, not SDP-certified optima.}
\label{fig:supp-performance-bandwidth}
\end{figure}

Within the tested cutoffs, codes, and measurement classes, we found no value satisfying
\begin{equation}
S>2+10^{-7}
\qquad
\text{for}
\qquad
\eta\leq\etaG.
\end{equation}
This negative result must be interpreted with care.  A see-saw search provides achieved values and therefore lower bounds on the best performance of the tested class.  It cannot exclude a better code or measurement that was not sampled, nor can it exclude convergence to a suboptimal stationary point.  The band-limited search has the additional limitation that its one-observable updates are projected approximations rather than exact SDP solutions.  Finally, the general-observable runs are permissive finite-dimensional benchmarks, not experimentally specified protocols.

Nevertheless, the results show that the golden-ratio boundary is not easily removed by shifting a single coherence, spreading the logical states across several Fock levels, increasing the measurement bandwidth, or allowing arbitrary finite-Fock observables within the tested dimensions.  This raises the possibility that \(\etaG\) reflects a broader operational boundary for recovery-free finite-Fock strategies, rather than an accidental feature of the vacuum-anchored family.  The present numerics do not establish such universality, and \(\etaG\) should not be interpreted as a no-go threshold beyond the family covered by Theorem~\ref{thm:supp-anchored-golden}.

Testing this possibility requires broader constructions.  Natural directions include multilevel codes with several deliberately aligned coherences, non-vacuum Fock-pair superpositions, binomial or cat-like bosonic encodings measured directly without recovery, exact SDP updates for band-limited observables, and joint optimization of the code isometry and measurements.  Such searches may either produce a counterexample below \(\etaG\) or identify additional structure behind a more general operational bound.

The finite-Fock search above asks how much can be gained by enlarging the local code and measurement spaces under pure loss.  We now take a complementary finite-size approach.  Instead of increasing the number of accessible Fock coherences, we introduce redundancy through a sparse binomial permutation-invariant code.  Its observables act directly on the survivor Dicke sectors, and its performance is established analytically for a controlled number of unflagged deletions.

\section{Sparse binomial PI construction}
\label{supp:sparse-proof}


\begin{definition}[Binomial PI code]
\label{def:supp-binomial-code}
Let \(\Mbin\) be odd, let \(g=t+1\), let \(n=g\Mbin\), and assume \(\Mbin>t\).  The binomial permutation-invariant code is
\begin{align}
\ket{0_L}
=2^{-(\Mbin-1)/2}
\sum_{\ell\ {\rm even}}
\sqrt{\binom{\Mbin}{\ell}}\,\ket{D_{g\ell}^{n}},
\qquad
\ket{1_L}
=2^{-(\Mbin-1)/2}
\sum_{\ell\ {\rm odd}}
\sqrt{\binom{\Mbin}{\ell}}\,\ket{D_{g\ell}^{n}} .
\label{eq:supp-binomial-code}
\end{align}
This is a PI Dicke-lattice encoding of the same general type as the PI deletion-code families studied in Refs.~\cite{Ouyang2014,Aydin2024,Shibayama2021}.  The logical label is the parity of the coarse index \(\ell\).
\end{definition}

On the survivor block \(\cH_{m,{\rm sym}}\), define
$a(q)=(-q)\bmod g$, and $L(q)=\frac{q+a(q)}{g}$.
The modular parity observable $Z_g^{(m)}$ and the transverse observable $\ContractX_g^{(m)}$ are respectively,
\begin{equation}
Z_g^{(m)}=
\sum_{q=0}^{m}(-1)^{L(q)}\ket{D_q^m}\!\bra{D_q^m},
\qquad
\ContractX_g^{(m)}=
\sum_{\substack{0\le q\le m-g\\ L(q)\ {\rm even}}}
\left(
\ket{D_q^m}\!\bra{D_{q+g}^m}
+
\ket{D_{q+g}^m}\!\bra{D_q^m}
\right).
\label{eq:supp-Zg}
\end{equation}
The unpaired subspace is assigned zero signed observable, i.e. randomized output.

\begin{lemma}[Partial Dicke-pair contraction]
\label{lem:supp-partial-shift-validity}
For every survivor block \(m\), \(\ContractX_g^{(m)}\) is Hermitian and satisfies \(\norm{\ContractX_g^{(m)}}_\infty\le1\).  It anticommutes with \(Z_g^{(m)}\) on the paired support and vanishes on the unpaired support.  Therefore
$B_{0,1}^{(m)}=\frac{Z_g^{(m)}\pm\ContractX_g^{(m)}}{\sqrt2}$
are Hermitian contractions and define valid binary POVMs through
\(E_\pm(B_y^{(m)})=(I\pm B_y^{(m)})/2\).
\end{lemma}

\begin{proof}
The summands in $\ContractX_g^{(m)}$ Eq.~\eqref{eq:supp-Zg} are Hermitian swaps between the two-dimensional subspaces
\(\operatorname{span}\{\ket{D_q^m},\ket{D_{q+g}^m}\}\).  These two-dimensional blocks are disjoint.  Indeed, if \(L(q)\) is even, then \(L(q+g)=L(q)+1\) is odd; hence \(q+g\) cannot also appear as the lower, even-labeled member of another pair.  Thus \(\ContractX_g^{(m)}\) is a direct sum of Pauli \(X\) matrices on the paired support and zero on its orthogonal complement.  This proves Hermiticity and \(\norm{\ContractX_g^{(m)}}_\infty\le1\).

On each paired block, \(Z_g^{(m)}\) has eigenvalues \(+1\) and \(-1\), because \(L(q+g)=L(q)+1\).  Therefore the restriction of \(Z_g^{(m)}\) is Pauli \(Z\), while the restriction of \(\ContractX_g^{(m)}\) is Pauli \(X\), and they anticommute.  On the unpaired support, \(\ContractX_g^{(m)}=0\).  Hence
\begin{equation}
\left(B_y^{(m)}\right)^2
=
\begin{cases}
I,&\text{on paired two-dimensional blocks},\\
\frac12 I,&\text{on the unpaired support},
\end{cases}
\end{equation}
so \(\norm{B_y^{(m)}}_\infty\le1\).  The effects \((I\pm B_y^{(m)})/2\) are positive and sum to the identity.  On the kernel of \(\ContractX_g^{(m)}\), the transverse signed observable is zero, which is exactly unbiased classical output randomization.\end{proof}

In the CHSH construction Alice's transverse setting is the signed observable \(\ContractX_g\); it gives unbiased randomized output on unpaired sectors.  Bob's tilted settings are \((Z_g\pm\ContractX_g)/\sqrt2\); on unpaired sectors these reduce to \(Z_g/\sqrt2\) and remain Hermitian contractions.

To evaluate the CHSH correlations after deletion, we now resolve the output into branches labeled by the number \(a\) of lost excitations.  The key point is that, for fewer than \(\Mbin\) deletions, each branch occurs with the same weight in the two logical sectors.  This balancing allows the post-deletion state to be written as a mixture of Bell-like branch states.
For exact deletion of \(r\) particles, recall
$Q_{\ell,a}^{(r)}=
\frac{\binom{g\ell}{a}\binom{g\Mbin-g\ell}{r-a}}{\binom{g\Mbin}{r}}$,
with $ 0\le a\le r$.
The binomial coefficient is understood to be zero outside its natural range.  The branch weights in the even and odd logical sectors are
\begin{align}
c_{0,a}^{(r)}&=2^{-(\Mbin-1)}
\sum_{\ell\ {\rm even}}\binom{\Mbin}{\ell}Q_{\ell,a}^{(r)},
&
 c_{1,a}^{(r)}&=2^{-(\Mbin-1)}
\sum_{\ell\ {\rm odd}}\binom{\Mbin}{\ell}Q_{\ell,a}^{(r)} .
\end{align}

\begin{lemma}[Finite-difference syndrome balancing]
\label{lem:supp-finite-diff-balancing}
For every \(r<\Mbin\) and every lost-excitation syndrome \(a\): $c_{0,a}^{(r)}=c_{1,a}^{(r)}=:c_a^{(r)}$.
\end{lemma}

\begin{proof}
The difference is
$c_{0,a}^{(r)}-c_{1,a}^{(r)}
=2^{-(\Mbin-1)}\sum_{\ell=0}^{\Mbin}
(-1)^\ell\binom{\Mbin}{\ell}Q_{\ell,a}^{(r)}$ .
As a function of \(\ell\), \(Q_{\ell,a}^{(r)}\) is a polynomial of degree at most \(r\): the factor \(\binom{g\ell}{a}\) has degree \(a\), and \(\binom{g\Mbin-g\ell}{r-a}\) has degree \(r-a\).  Since \(r<\Mbin\), the \(\Mbin\)-th finite difference of this polynomial vanishes:
$\sum_{\ell=0}^{\Mbin}(-1)^\ell\binom{\Mbin}{\ell}P(\ell)=0$ with $\deg P<\Mbin$.
This proves the equality of the even and odd branch weights.\end{proof}

For \(r\le t=g-1\), the syndrome \(a\) is encoded in the residue class of the surviving Dicke weight: if \(q=g\ell-a\), then \(a(q)=a\) and \(L(q)=\ell\).  Define the normalized branch states
\begin{align}
\ket{\widetilde 0_a^{(r)}}
=\frac{1}{\sqrt{2^{(\Mbin-1)}c_a^{(r)}}}
\sum_{\substack{\ell\ {\rm even}\\ g\ell\ge a}}
\sqrt{\binom{\Mbin}{\ell}Q_{\ell,a}^{(r)}}\,
\ket{D_{g\ell-a}^{n-r}},
\quad
\ket{\widetilde 1_a^{(r)}}
=\frac{1}{\sqrt{2^{(\Mbin-1)}c_a^{(r)}}}
\sum_{\substack{\ell\ {\rm odd}\\ g\ell\ge a}}
\sqrt{\binom{\Mbin}{\ell}Q_{\ell,a}^{(r)}}\,
\ket{D_{g\ell-a}^{n-r}} .
\end{align}
If a coefficient has zero norm, the corresponding normalized vector is omitted from the branch mixture; it has zero weight.

For fixed deletion numbers \(r_A,r_B\le t\), the post-deletion encoded Bell state decomposes as
\begin{equation}
\rho_{r_A,r_B}=
\sum_{a=0}^{r_A}\sum_{b=0}^{r_B}
 c_a^{(r_A)}c_b^{(r_B)}
\ket{\Phi_{a,b}^{(r_A,r_B)}}\!\bra{\Phi_{a,b}^{(r_A,r_B)}},
\quad
\ket{\Phi_{a,b}^{(r_A,r_B)}}=
\frac{
\ket{\widetilde0_a^{(r_A)}}_A\ket{\widetilde0_b^{(r_B)}}_B+
\ket{\widetilde1_a^{(r_A)}}_A\ket{\widetilde1_b^{(r_B)}}_B}{\sqrt2}.
\label{eq:supp-branch-mixture}
\end{equation}
This follows directly from the exact-\(r\) Kraus maps and Lemma~\ref{lem:supp-finite-diff-balancing}.  The branch weights are normalized because
\begin{equation}
\sum_{a=0}^{r} c_a^{(r)}
=2^{-(\Mbin-1)}\sum_{\ell\ {\rm even}}\binom{\Mbin}{\ell}
\sum_{a=0}^{r} Q_{\ell,a}^{(r)}
=2^{-(\Mbin-1)}\sum_{\ell\ {\rm even}}\binom{\Mbin}{\ell}=1,
\end{equation}
and the same calculation holds on the odd sector.  Hence \(\sum_{a,b}c_a^{(r_A)}c_b^{(r_B)}=1\).  The unnormalized \((a,b)\) branch of the logical Bell state has squared norm \(c_a^{(r_A)}c_b^{(r_B)}\), since the even and odd contributions have equal branch weights.

\begin{lemma}[Branch Bell decomposition]
\label{lem:supp-branch-bell-decomposition}
For each good branch \(r\le t\) and syndrome \(a\),
\begin{equation}
Z_g\ket{\widetilde0_a^{(r)}}=\ket{\widetilde0_a^{(r)}},
\qquad
Z_g\ket{\widetilde1_a^{(r)}}=-\ket{\widetilde1_a^{(r)}} .
\end{equation}
Moreover,
\begin{equation}
\mu_a^{(r)}:=
\bra{\widetilde0_a^{(r)}}\ContractX_g\ket{\widetilde1_a^{(r)}}
=
\frac{1}{c_a^{(r)}2^{\Mbin-1}}
\sum_{\ell\ {\rm even}}
\sqrt{\binom{\Mbin}{\ell}\binom{\Mbin}{\ell+1}
Q_{\ell,a}^{(r)}Q_{\ell+1,a}^{(r)}} .
\label{eq:supp-mu-ra}
\end{equation}
\end{lemma}

\begin{proof}
For a good branch, \(g\ell-a\) has syndrome \(a\) and reconstructed index \(L(g\ell-a)=\ell\).  Thus Eq.~\eqref{eq:supp-Zg} returns the parity of \(\ell\), proving the first statement.  The contraction \(\ContractX_g\) pairs only Dicke weights separated by \(g\), so it connects \(g\ell-a\) to \(g(\ell+1)-a\) and no non-neighboring logical sectors.  Substituting the normalized branch expansions gives Eq.~\eqref{eq:supp-mu-ra}.  
The sum may formally include boundary values of \(\ell\) for which one of the two Dicke weights,
$g\ell-a$ or $g(\ell+1)-a$, lies outside the physical range \(0,\ldots,n-r\).  In that case the corresponding coefficient \(Q_{\ell,a}^{(r)}\) or \(Q_{\ell+1,a}^{(r)}\) is zero, by the convention that binomial coefficients vanish outside their natural range.  The associated summand therefore vanishes automatically.  Hence Eq.~\eqref{eq:supp-mu-ra} includes every admissible boundary pair without requiring separate the first and the last possible values of $\ell$ terms.
\end{proof}
\begin{proposition}[Good-branch correlators]
\label{prop:supp-fixed-branch-correlators}
For exact deletion numbers \(r_A,r_B\le t\), the partial-contraction observables satisfy
\begin{equation}
\langle Z_g\otimes Z_g\rangle=1,
\qquad
\langle \ContractX_g\otimes \ContractX_g\rangle=\mu_{r_A}\mu_{r_B},
\qquad
\langle Z_g\otimes\ContractX_g\rangle=
\langle\ContractX_g\otimes Z_g\rangle=0 .
\end{equation}
With Alice's observables \(A_0=Z_g\), \(A_1=\ContractX_g\), and Bob's signed contractions \(B_{0,1}=(Z_g\pm\ContractX_g)/\sqrt2\), the good-branch CHSH value is
\begin{equation}
S_{r_A,r_B}=\sqrt2\left(1+\mu_{r_A}\mu_{r_B}\right),
\qquad \text{with }\;
\mu_r=\sum_{a=0}^{r}c_a^{(r)}\mu_a^{(r)}.
\end{equation}
\end{proposition}

\begin{proof}
For every branch Bell state in Eq.~\eqref{eq:supp-branch-mixture}, \(Z_g\) acts as logical Pauli \(Z\), hence \(\langle Z_g\otimes Z_g\rangle=1\).  The mixed correlators vanish because \(Z_g\) is diagonal in the branch logical basis whereas \(\ContractX_g\) has only off-diagonal matrix elements between the two branch logical states.  Finally, Lemma~\ref{lem:supp-branch-bell-decomposition} gives the transverse branch contribution \(\mu_a^{(r_A)}\mu_b^{(r_B)}\).  Averaging over independent syndrome weights gives \(\mu_{r_A}\mu_{r_B}\).  Substitution into the CHSH expression yields the stated value.
\end{proof}
We have now completed the structural part of the construction.  For every deletion pattern with at most \(t\) losses, the output decomposes into balanced Bell-like branches, and the corresponding CHSH value is determined by the transverse overlaps \(\mu_r\).  The following theorem summarizes these results and identifies the remaining step: obtaining uniform lower bounds on these overlaps and averaging them over the stochastic deletion process.
\begin{theorem}[Sparse PI theorem, structural form]
\label{thm:supp-sparse-pi-structural}
For the code in Eq.~\eqref{eq:supp-binomial-code} and the partial-contraction observables in Eqs.~\eqref{eq:supp-Zg}, every good exact-deletion branch \(r_A,r_B\leq t\) has the CHSH value given in Proposition~\ref{prop:supp-fixed-branch-correlators}.  Branches with more than \(t\) deletions are assigned randomized signed outputs.  Combining this structure with the estimates derived below gives the finite-size lower bound for independent stochastic deletion stated in Theorem~\ref{thm:supp-sparse-stochastic}.
\end{theorem}
It remains to control the transverse overlaps \(\mu_r\) uniformly over all good deletion branches.  These overlaps compare the loss distributions of neighboring coarse Dicke sectors.  We quantify their similarity through the corresponding Hellinger affinity and then use the resulting estimate to treat independent stochastic deletion.

For neighboring coarse Dicke sectors, define the hypergeometric Hellinger affinity
\begin{equation}
h_\ell^{(r)}=
\sum_{a=0}^{r}\sqrt{Q_{\ell,a}^{(r)}Q_{\ell+1,a}^{(r)}} .
\label{eq:supp-h-ell-r}
\end{equation}
We lower-bound this quantity using ordered sampling without replacement.  For a starting population with \(g\ell\) excitations among \(g\Mbin\) objects, let \(P_\ell(x_1,\ldots,x_r)\) be the probability of an ordered deletion history, with \(x_s=1\) denoting that the \(s\)-th deleted particle is an excitation.  The count map \((x_1,\ldots,x_r)\mapsto a=\sum_s x_s\) sends \(P_\ell\) to \(Q_{\ell,a}^{(r)}\).

Hellinger affinity is monotone under coarse graining.  Therefore
\begin{equation}
h_\ell^{(r)}
\ge
\sum_{x_1,\ldots,x_r}\sqrt{P_\ell(x_1,
\ldots,x_r)P_{\ell+1}(x_1,\ldots,x_r)} .
\label{eq:supp-coarse-grain-hellinger}
\end{equation}
The inequality follows from Cauchy--Schwarz applied within each count class.  We do not identify the two affinities by equality.

\begin{lemma}[Bernoulli Hellinger step]
\label{lem:supp-bernoulli-hellinger-step}
For any \(p,p'\in[1/4,3/4]\), the Bernoulli Hellinger affinity obeys
\begin{equation}
\sqrt{pp'}+\sqrt{(1-p)(1-p')}\ge 1-(p-p')^2 .
\label{eq:supp-bernoulli-step}
\end{equation}
\end{lemma}

\begin{proof}
Let \(\delta=p-p'\).  The squared Hellinger distance identity gives
\begin{equation}
1-\sqrt{pp'}-\sqrt{(1-p)(1-p')}
=\frac12\left[(\sqrt p-\sqrt{p'})^2+(\sqrt{1-p}-\sqrt{1-p'})^2\right].
\end{equation}
Since \(p,p',1-p,1-p'\ge1/4\), one has \(\sqrt p+\sqrt{p'}\ge1\) and \(\sqrt{1-p}+\sqrt{1-p'}\ge1\).  Therefore
\begin{equation}
(\sqrt p-\sqrt{p'})^2=\frac{\delta^2}{(\sqrt p+\sqrt{p'})^2}\le\delta^2,
\qquad
(\sqrt{1-p}-\sqrt{1-p'})^2\le\delta^2.
\end{equation}
Substituting these two estimates proves Eq.~\eqref{eq:supp-bernoulli-step}.
\end{proof}

\begin{lemma}[Central-window Hellinger bound]
\label{lem:supp-hellinger-lower}
Let \(\Mbin\ge24\), \(r\le t=g-1\), and let even \(\ell\) satisfy
\begin{equation}
\frac{3\Mbin}{8}\le \ell\le \frac{5\Mbin}{8}-1 .
\end{equation}
Then
\begin{equation}
h_\ell^{(r)}\ge
\left(1-\frac{1}{(\Mbin-1)^2}\right)^r
\ge
\exp\left[-\frac{2r}{(\Mbin-1)^2}\right].
\label{eq:supp-h-lower}
\end{equation}
\end{lemma}

\begin{proof}
Fix a common ordered history of length \(s<r\) with \(j\) previous successes.  The next-step conditional success probabilities for the two neighboring populations are
\begin{equation}
p_s=\frac{g\ell-j}{g\Mbin-s},
\qquad
p_s'=\frac{g(\ell+1)-j}{g\Mbin-s} .
\end{equation}
For \(r=0\) there is nothing to prove.  For \(r\ge1\), \(g\ge2\) and \(0\le s\le r-1\le g-2\).  Since \(j\le s\), the central-window assumptions imply
\begin{equation}
\frac14\le p_s,p_s'\le\frac34,
\qquad
|p_s-p_s'|=\frac{g}{g\Mbin-s}\le\frac{1}{\Mbin-1} .
\end{equation}
The first inequality follows, for example, from
\(p_s\ge 3/8-1/\Mbin>1/4\) and
\(p_s'\le (5\Mbin/8)/(\Mbin-1)<3/4\) for \(\Mbin\ge24\); the remaining cases are easier.

By Lemma~\ref{lem:supp-bernoulli-hellinger-step}, each conditional one-step affinity is at least
\begin{equation}
\gamma:=1-\frac{1}{(\Mbin-1)^2}.
\end{equation}
Let \(H_s\) be the ordered-history affinity after \(s\) steps.  Recursively,
\begin{equation}
H_{s+1}
=\sum_{h_s}\sqrt{P_\ell(h_s)P_{\ell+1}(h_s)}
\sum_{x\in\{0,1\}}\sqrt{P_\ell(x|h_s)P_{\ell+1}(x|h_s)}
\ge \gamma H_s .
\end{equation}
Since \(H_0=1\), the ordered-history affinity is at least \(\gamma^r\).  By Eq.~\eqref{eq:supp-coarse-grain-hellinger}, the hypergeometric count affinity is no smaller.  Finally, \(1-x\ge e^{-2x}\) for \(0\le x\le1/2\), and here \(x=(\Mbin-1)^{-2}\le1/23^2\).
\end{proof}

The Hellinger estimate controls how much the deletion statistics change between neighboring coarse Dicke sectors.  We now translate this statistical overlap into a lower bound on the branch coherence \(\mu_r\), which determines the transverse correlator and hence the CHSH value in Proposition~\ref{prop:supp-fixed-branch-correlators}.  After averaging Eq.~\eqref{eq:supp-mu-ra} over the lost-excitation syndrome \(a\), the coherence becomes
\begin{equation}
\mu_r
=
2^{-(\Mbin-1)}
\sum_{\ell\ {\rm even}}
\sqrt{
\binom{\Mbin}{\ell}
\binom{\Mbin}{\ell+1}
}
\,h_\ell^{(r)} .
\end{equation}
Let \(\mathcal I_\Mbin\) be the even indices in the central window
\(3\Mbin/8\le\ell\le5\Mbin/8-1\).  Lemma~\ref{lem:supp-hellinger-lower} gives
\begin{equation}
\mu_r\ge
\exp\left[-\frac{2r}{(\Mbin-1)^2}\right]B_\Mbin(\mathcal I_\Mbin),
\end{equation}
where \(B_\Mbin(\mathcal I_\Mbin)\) is the corresponding binomial pairing mass.

\begin{lemma}[Full even pairing affinity]
\label{lem:supp-full-even-pairing}
Let \(\Mbin\) be odd.  Then
\begin{equation}
2^{-(\Mbin-1)}
\sum_{\ell\ {\rm even}}
\sqrt{\binom{\Mbin}{\ell}\binom{\Mbin}{\ell+1}}
\ge
1-2\frac{\binom{\Mbin}{\lfloor \Mbin/2\rfloor}}{2^\Mbin} .
\label{eq:supp-full-even-pairing}
\end{equation}
The sum is over even \(\ell\) for which \(0\le \ell\le \Mbin-1\).
\end{lemma}

\begin{proof}
Let \(\Omega_{\rm ev}=\{\ell:0\le\ell\le\Mbin-1,\ \ell\ {\rm even}\}\).  Since \(\Mbin\) is odd, the even and odd binomial masses are both \(2^{\Mbin-1}\).  Define probability distributions on \(\Omega_{\rm ev}\) by
\begin{equation}
P_\ell=2^{-(\Mbin-1)}\binom{\Mbin}{\ell},
\qquad
Q_\ell=2^{-(\Mbin-1)}\binom{\Mbin}{\ell+1}.
\end{equation}
They are normalized because \(\ell\) runs over even indices and \(\ell+1\) runs over odd indices.  The left side of Eq.~\eqref{eq:supp-full-even-pairing} is the Hellinger affinity \(F(P,Q)=\sum_{\ell\in\Omega_{\rm ev}}\sqrt{P_\ell Q_\ell}\).
For any two probability distributions, the affinity and total variation obey
\begin{equation}
F(P,Q)\ge 1-\frac12\sum_{\ell\in\Omega_{\rm ev}}|P_\ell-Q_\ell|.
\label{eq:supp-affinity-tv}
\end{equation}
Indeed, \(1-F(P,Q)=\frac12\sum_\ell(\sqrt{P_\ell}-\sqrt{Q_\ell})^2\le \frac12\sum_\ell|P_\ell-Q_\ell|\), because \((\sqrt x-\sqrt y)^2\le |x-y|\) for \(x,y\ge0\).
It remains to bound the variation.  Put \(a_j=\binom{\Mbin}{j}\) and \(h=\lfloor\Mbin/2\rfloor\).  The binomial sequence is increasing up to \(h\), has the equal central maxima \(a_h=a_{h+1}\), and is decreasing after \(h+1\).  Hence the full adjacent variation telescopes:
\begin{equation}
\sum_{j=0}^{\Mbin-1}|a_{j+1}-a_j|
= (a_h-a_0)+(a_{h+1}-a_{\Mbin})
=2(a_h-1)\le 2a_h.
\label{eq:supp-binomial-telescope}
\end{equation}
The even-pair sum is a sub-sum of this full adjacent variation, so \(\sum_{\ell\in\Omega_{\rm ev}}|a_\ell-a_{\ell+1}|\le2a_h\).  Hence
\begin{equation}
\frac12\sum_{\ell\in\Omega_{\rm ev}}|P_\ell-Q_\ell|
=2^{-\Mbin}\sum_{\ell\in\Omega_{\rm ev}}|a_\ell-a_{\ell+1}|
\le 2\frac{a_h}{2^\Mbin}.
\end{equation}
Combining this with Eq.~\eqref{eq:supp-affinity-tv} proves the lemma.
\end{proof}

By Lemma~\ref{lem:supp-full-even-pairing}, the full even pairing affinity obeys Eq.~\eqref{eq:supp-full-even-pairing}.

For the tails, AM--GM gives
\begin{equation}
\sqrt{\binom{\Mbin}{\ell}\binom{\Mbin}{\ell+1}}
\le \frac12\left[\binom{\Mbin}{\ell}+\binom{\Mbin}{\ell+1}\right].
\end{equation}
After relaxing the parity restriction, Hoeffding's inequality for a fair binomial random variable gives the tail bound
\begin{equation}
2^{-(\Mbin-1)}
\sum_{\substack{\ell\ {\rm even}\\ \ell\notin\mathcal I_{\Mbin}}}
\sqrt{\binom{\Mbin}{\ell}\binom{\Mbin}{\ell+1}}
\le 4e^{-\Mbin/32} .
\end{equation}

\begin{proposition}[Uniform coherence lower bound]
\label{prop:supp-mu-star-bound}
For every \(r\le t\),
\begin{equation}
\mu_r\ge\mu_*:=
\exp\left[-\frac{2t}{(\Mbin-1)^2}\right]
\left[
1-
2\frac{\binom{\Mbin}{\lfloor\Mbin/2\rfloor}}{2^{\Mbin}}
-
4e^{-\Mbin/32}
\right]_+ .
\label{eq:supp-mu-star}
\end{equation}
\end{proposition}

\begin{proof}
Use the previous central-window estimate and replace \(r\) by its maximum \(t\).  The positive part is harmless and ensures a nonnegative lower bound even outside the useful parameter region.\end{proof}

The analysis so far applies to every fixed deletion number \(r\leq t\), for which Proposition~\ref{prop:supp-mu-star-bound} provides a uniform lower bound on the surviving coherence.  We now average these fixed-\(r\) results over independent stochastic deletion.  The probability of leaving the good region \(r\leq t\) is controlled by a binomial tail bound, while branches with more than \(t\) deletions are assigned randomized signed outputs.  This yields a finite-size CHSH bound and, for the square family \(n=\Mbin^2\), an explicit loss scaling that guarantees violation.

For independent deletion, \(\Rdel\sim\Bin(n,p)\).  If \(np<t\), Bernstein's inequality gives
\begin{equation}
\Pr[\Rdel>t]\le
\epsq:=
\exp\left[-\frac{(t-np)^2}{2np(1-p)+\frac23(t-np)}\right].
\label{eq:supp-eps-q}
\end{equation}
Let \(q_t=\Pr[\Rdel\le t]\).  Good branches use Proposition~\ref{prop:supp-fixed-branch-correlators}; bad branches are assigned randomized signed outputs.  Thus the stochastic \(Z_gZ_g\) contribution is at least \(q_t^2\), and the stochastic transverse contribution is at least
\begin{equation}
\nu_t^2,
\qquad
\nu_t:=\sum_{r=0}^{t}\Pr[\Rdel=r]\mu_r\ge q_t\mu_* .
\end{equation}
Therefore
\begin{theorem}[Stochastic sparse-deletion lower bound]
\label{thm:supp-sparse-stochastic}
Let \(\Mbin\) be odd, \(\Mbin\ge24\), \(g=t+1\), \(n=g\Mbin\), and \(\Mbin>t\).  If \(np<t\), the sparse PI protocol satisfies
\begin{equation}
S(n,p)
\ge
\sqrt2(1-\epsq)^2(1+\mu_*^2),
\label{eq:supp-sparse-bound}
\end{equation}
with \(\epsq\) and \(\mu_*\) as in Eqs.~\eqref{eq:supp-eps-q} and \eqref{eq:supp-mu-star}.
\end{theorem}

\begin{proof}
Since \(q_t\ge1-\epsq\),
\begin{equation}
S(n,p)\ge\sqrt2(q_t^2+\nu_t^2)
\ge\sqrt2(1-\epsq)^2(1+\mu_*^2).
\end{equation}
The first inequality is the CHSH expression of Proposition~\ref{prop:supp-fixed-branch-correlators} averaged over independent deletion numbers, with zero signed contribution assigned to bad branches.\end{proof}

\begin{corollary}[Square-family consequence]
\label{cor:supp-square-family}
Let
$g=\Mbin$, $t=\Mbin-1$, $n=\Mbin^2$,
with odd \(\Mbin\ge101\).  If
\begin{equation}
p\le\frac{1}{4\Mbin}=\frac{1}{4\sqrt n},
\end{equation}
then the lower bound in Theorem~\ref{thm:supp-sparse-stochastic} is strictly larger than \(2\).  Hence the sparse PI protocol certifies recovery-free CHSH violation for \(p=O(n^{-1/2})\).
\end{corollary}

\begin{proof}
It is enough to check the largest allowed value \(p=1/(4\Mbin)\), because decreasing \(p\) decreases the binomial upper-tail probability and leaves \(\mu_*\) unchanged.  At \(p=1/(4\Mbin)\), the binomial mean is \(np=\Mbin/4\), and
\begin{equation}
t-np=\frac{3\Mbin}{4}-1,
\qquad
2np(1-p)+\frac23(t-np)=\Mbin-\frac{19}{24}.
\end{equation}
The Bernstein exponent is
\begin{equation}
E(\Mbin)=\frac{(3\Mbin/4-1)^2}{\Mbin-19/24}.
\end{equation}
A direct derivative gives
\begin{equation}
E'(\Mbin)=\frac{27(3\Mbin-4)(4\Mbin-1)}{(24\Mbin-19)^2}>0
\qquad (\Mbin\ge101),
\end{equation}
and \(E(101)=20631/370>55\).  Thus \(\epsq\le e^{-55}\) for all \(\Mbin\ge101\).

For the coherence, use the standard central-binomial estimate
$\frac{\binom{\Mbin}{\lfloor\Mbin/2\rfloor}}{2^\Mbin}
\le \sqrt{\frac{2}{\pi\Mbin}}$.
For \(\Mbin\ge101\), the elementary monotonicities of \(e^{-2/(\Mbin-1)}\), \(-\sqrt{2/(\pi\Mbin)}\), and \(-e^{-\Mbin/32}\) give the uniform lower bound
\begin{equation}
\mu_*\ge
\exp\left[-\frac{2}{100}\right]
\left[1-2\sqrt{\frac{2}{101\pi}}-4e^{-101/32}\right]>0.65.
\end{equation}
Consequently
\begin{equation}
S(n,p)
\ge
\sqrt2(1-e^{-55})^2(1+0.65^2)>2 .
\end{equation}
This proves the corollary for all odd \(\Mbin\ge101\).
\end{proof}
\subsection{Numerical results}
\label{supp:sparse-numerics}

The analytical bound of Theorem~\ref{thm:supp-sparse-stochastic} can be evaluated directly to illustrate how the sparse binomial PI construction behaves away from the particular parameter choice used in Corollary~\ref{cor:supp-square-family}.  We focus on the square family
\begin{equation}
g=\Mbin,
\qquad
t=\Mbin-1,
\qquad
n=\Mbin^2,
\label{eq:supp-square-family-numerics}
\end{equation}
and introduce the scaled deletion parameter
\begin{equation}
c=p\Mbin,
\qquad
p=\frac{c}{\Mbin},\qquad p=1-\eta.
\label{eq:supp-scaled-deletion-c}
\end{equation}
For fixed \(c=O(1)\), this is precisely the sparse-deletion scaling
\[
p=O(\Mbin^{-1})=O(n^{-1/2})
\]
established in the theorem.  The corresponding certified CHSH lower bound is
\begin{equation}
S_{\rm lb}(\Mbin,c)
=
\sqrt2\,
\bigl(1-\epsq(\Mbin,c)\bigr)^2
\bigl(1+\mu_*^2(\Mbin)\bigr),
\label{eq:supp-square-Slb}
\end{equation}
where \(\epsq\) is obtained from Eq.~\eqref{eq:supp-eps-q} after setting \(p=c/\Mbin\), while \(\mu_*\) is the uniform coherence bound of Eq.~\eqref{eq:supp-mu-star} with \(t=\Mbin-1\).

\begin{figure}[t]
\centering
\begin{minipage}[t]{0.49\linewidth}
    \centering
    \includegraphics[width=\linewidth]{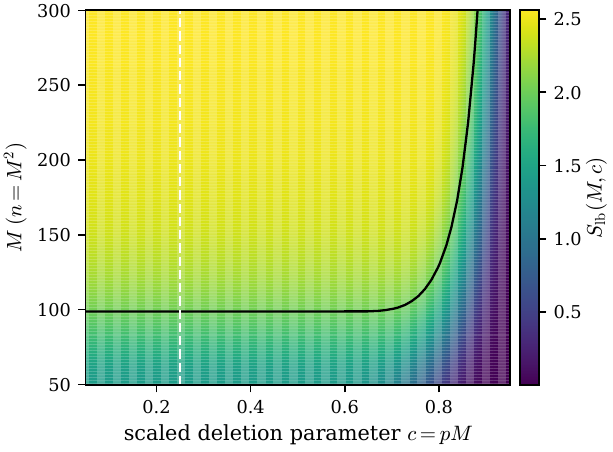}
    \vspace{1mm}
    \centerline{\small (a)}
\end{minipage}
\hfill
\begin{minipage}[t]{0.49\linewidth}
    \centering
    \includegraphics[width=\linewidth]{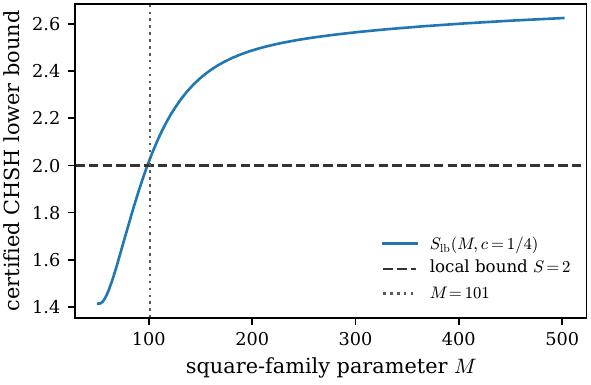}
    \vspace{1mm}
    \centerline{\small (b)}
\end{minipage}
\caption{
Numerical evaluation of the analytical lower bound for the square-family sparse binomial PI construction of Sec.~\ref{supp:sparse-proof}.
(a) Certified lower bound \(S_{\rm lb}(\Mbin,c)\) as a function of the square-family parameter \(\Mbin\), with \(n=\Mbin^2\), and the scaled deletion probability \(c=p\Mbin\).  The black contour marks \(S_{\rm lb}=2\), so points on the side of the contour with \(S_{\rm lb}>2\) are certified by Theorem~\ref{thm:supp-sparse-stochastic} to violate CHSH without recovery.  The vertical dashed line marks the choice \(c=1/4\), corresponding to \(p=1/(4\Mbin)\) used in Corollary~\ref{cor:supp-square-family}.
(b) Cross section of the same lower bound along \(c=1/4\).  The dashed horizontal line is the local CHSH bound \(S=2\), and the vertical dotted line marks \(\Mbin=101\), from which the corollary guarantees violation for the allowed odd values of \(\Mbin\).  Together, the two panels show how increasing the PI redundancy improves the certified transverse coherence while suppressing the probability of leaving the good-deletion sector.
}
\label{fig:supp-sparse-numerical-results}
\end{figure}
Panel~(a) gives a two-parameter view of the sparse-deletion region.  The horizontal variable \(c\) is useful because it separates the scaling with system size from the numerical prefactor of the deletion probability.  Indeed, since \(n=\Mbin^2\),
$np
=
\Mbin^2\frac{c}{\Mbin}
=
c\Mbin$ ,
whereas the largest deletion number treated by the good-branch analysis is
$t=\Mbin-1.$
Thus, for fixed \(c<1\) and increasing \(\Mbin\), the mean number of deleted particles and the correctable deletion scale both grow linearly with \(\Mbin\), while their ratio remains controlled.  The figure should therefore not be interpreted as showing tolerance to a fixed nonzero deletion probability as \(n\) grows.  Rather, it displays the finite-size behavior inside the sparse regime \(p=c/\Mbin\).

The color variation in panel~(a) reflects the two factors in Eq.~\eqref{eq:supp-square-Slb}.  The first is the good-branch probability, controlled by \(\epsq\).  For fixed \(c\) sufficiently below unity, the binomial mean \(np=c\Mbin\) remains separated from the cutoff \(t=\Mbin-1\), and the probability of observing more than \(t\) deletions decreases rapidly as \(\Mbin\) grows.  The second factor is the coherence term \(\mu_*\).  In the square family,
\begin{equation}
\mu_*
=
\exp\!\left[-\frac{2}{\Mbin-1}\right]
\left[
1-
2\frac{\binom{\Mbin}{\lfloor\Mbin/2\rfloor}}{2^\Mbin}
-
4e^{-\Mbin/32}
\right]_+ ,
\end{equation}
so the lower bound on the surviving Dicke-pair coherence improves with increasing \(\Mbin\) and approaches unity asymptotically.  These two effects explain the broad high-\(S_{\rm lb}\) region at large \(\Mbin\): the stochastic deletion process remains with high probability inside the syndrome-resolved sector \(r\le t\), while neighboring coarse Dicke sectors become increasingly difficult to distinguish through their deletion statistics.

The black contour in panel~(a) gives a direct numerical representation of the sufficient condition
\[
S_{\rm lb}(\Mbin,c)>2.
\]
It is important that this contour is a boundary of the \emph{certified lower bound}, not an exact phase boundary for the optimal CHSH value.  Points for which the plotted lower bound is below \(2\) are simply not certified by Theorem~\ref{thm:supp-sparse-stochastic}; the calculation does not imply that the corresponding post-deletion state is local.  The sharp upward bending of the contour at larger \(c\) has a simple interpretation.  As \(c\) approaches unity, the mean number of deletions \(np=c\Mbin\) approaches the maximal good-branch value \(t=\Mbin-1\).  The safety margin entering the Bernstein exponent,
\begin{equation}
t-np
=
(1-c)\Mbin-1,
\end{equation}
then decreases, so bad branches can no longer be neglected at moderate system size.  Increasing \(\Mbin\) is consequently required to recover a sufficiently small \(\epsq\) and hence a certified violation.

The dashed line \(c=1/4\) selects the explicit scaling used in Corollary~\ref{cor:supp-square-family}.  Along this line,
\begin{equation}
p=\frac{1}{4\Mbin}
=
\frac{1}{4\sqrt n},
\qquad
np=\frac{\Mbin}{4},
\end{equation}
so the mean deletion number remains well below \(t=\Mbin-1\).  Panel~(b) isolates this cross section and makes the finite-size onset of the analytical certification transparent.  The lower bound increases with \(\Mbin\), crosses the local value \(2\) at the scale highlighted by \(\Mbin=101\), and then continues to improve.  This is the numerical counterpart of Corollary~\ref{cor:supp-square-family}, which proves that all allowed odd \(\Mbin\ge101\) satisfy the CHSH-violation condition for \(p\le1/(4\Mbin)\).

The growth in panel~(b) should again be understood through the two ingredients of the proof.  The factor \((1-\epsq)^2\) approaches unity because stochastic deletion increasingly concentrates inside the good region, while \(\mu_*^2\) increases because the Hellinger overlap between neighboring coarse Dicke sectors becomes closer to one.  Consequently, the direct observables \(Z_g\) and \(\ContractX_g\) retain increasingly ideal logical correlations even though no recovery operation is performed.  In the large-\(\Mbin\) limit at fixed \(c<1\), the analytical lower bound approaches the ideal CHSH value \(2\sqrt2\); at the same time the physical deletion probability itself decreases as \(1/\Mbin\).

These numerical results therefore complement the proof rather than introduce a separate optimization.  They show explicitly how the finite-difference syndrome balancing, the Hellinger coherence estimate, and the stochastic tail bound combine at finite size.  In particular, they make clear that the sparse binomial PI construction is a redundancy mechanism: increasing \(\Mbin\) enlarges the separation between the typical deletion count and the failure region while preserving the logical parity information in the survivor Dicke lattice.  This is the finite-size behavior underlying the \(p=O(n^{-1/2})\) recovery-free violation stated in the main text.
\section{Experimental implementation notes}
\label{supp:experiment}

The measurement models used in the Letter can be implemented with existing state-preparation, number-resolved detection, and coherence-control techniques, although their difficulty varies considerably across the three protocols.  The discussion below identifies the required experimental operations and the main trade-offs.  It does not claim a complete implementation or a platform-specific error budget.  Its purpose is to connect the analytical thresholds in the main text with measurable quantities and to clarify which protocol is most natural in each experimental regime.

A test would begin by preparing two local systems in the encoded Bell state \(\ket{\Phi_L^+}\), with one encoded subsystem sent to Alice and the other to Bob.  Each subsystem would then undergo a calibrated loss channel with survival probability \(\eta\).  Alice and Bob would independently choose one of two binary measurement settings and record an outcome \(a,b\in\{+1,-1\}\) on every trial.  No event should be discarded because of the observed particle number or loss outcome.  Outcomes on unused measurement subspaces must instead be assigned according to the randomized-output rules specified in the protocols.  From the four correlators
$E_{xy}
=
\sum_{a,b=\pm1}
ab\,P(a,b|x,y)$,
 with $x,y\in\{0,1\}$,
one obtains
$S=E_{00}+E_{01}+E_{10}-E_{11}$.
The survival probability can be calibrated from number-resolved measurements, while the effective transverse visibility \(v\) should be determined independently by detector tomography or by applying the coherence measurement to known superposition states.  Scanning \(\eta\) would then test the threshold hierarchy described in the main text: the one-excitation protocol above \(1/\sqrt2\), the anchored family above \(\etaG\), and the universal absence of CHSH violation at or below \(1/2\).

The one-excitation protocol is the simplest starting point.  It requires preparation and control in the logical subspace
$\operatorname{span}
\left\{
\ket{D_0^n},
\ket{D_1^n}
\right\}$,
followed after loss by measurements in each survivor-number sector
\(\operatorname{span}\{\ket{D_0^m},\ket{D_1^m}\}\).  The longitudinal observable is obtained from population readout: it distinguishes the vacuum from the one-excitation state.  The transverse observable must instead measure the relative phase between these states.  Experimentally, this can be implemented by a calibrated rotation between \(\ket{D_0^m}\) and \(\ket{D_1^m}\), followed by the same population measurement.  The tilted settings entering CHSH are produced by changing the rotation angle, or more generally by implementing the binary POVM associated with
$O_v(\theta)=\cos\theta\,Z+v\sin\theta\,X$.
Because the survivor number \(m\) is part of the local output space, the rotation may be conditioned on \(m\).  This does not constitute recovery or postselection: every survivor sector is measured and every trial contributes a binary outcome.

The main advantage of the one-excitation encoding is its low excitation cost.  Its logical states and number measurement are comparatively simple, and loss cannot populate levels above the one-excitation sector.  It is therefore a natural choice when state preparation is the dominant limitation or when only low-order coherent control is available.  Its disadvantage is the threshold \(\eta>1/\sqrt2\), which is higher than the anchored asymptotic threshold \(\etaG\).  The protocol is also sensitive to the quality of the vacuum--one-excitation coherence measurement.  Population readout alone is insufficient, because it contains no information about the off-diagonal element
\(\ket{D_0^m}\!\bra{D_1^m}\).

A global collective Ramsey pulse offers a simple approximation to this transverse measurement, but it is not an exact logical Hadamard for \(m>1\).  The collective raising operator satisfies
\begin{equation}
J_+\ket{D_q^m}
=
\sqrt{(q+1)(m-q)}\,
\ket{D_{q+1}^m}.
\end{equation}
A pulse that couples \(\ket{D_0^m}\) to \(\ket{D_1^m}\) therefore also couples \(\ket{D_1^m}\) to \(\ket{D_2^m}\), producing leakage from the logical subspace.  If a Ramsey rotation is followed by Dicke-number binning, a standard \(\beta=\pi/2\) pulse gives the transverse contrast
\begin{equation}
\lambda_m
=
\frac{1}{\sqrt m\,2^m}
\sum_{k=0}^{m}
\binom{m}{k}|2k-m|,
\label{eq:supp-ramsey-contrast}
\end{equation}
or equivalently
\begin{equation}
\lambda_m=
\begin{cases}
\sqrt m\binom{m}{m/2}/2^m,
& m\ {\rm even},\\[2pt]
(m+1)\binom{m}{(m-1)/2}/(\sqrt m\,2^m),
& m\ {\rm odd}.
\end{cases}
\end{equation}
Stirling's formula gives
\(\lambda_m\to\sqrt{2/\pi}\).  A global Ramsey measurement is thus easy to apply and does not require selective addressing of the logical transition, but it realizes an unsharp transverse observable with a visibility that remains below one.  It is preferable when collective control is accurate and rapid, whereas a selective logical rotation is preferable when the larger visibility compensates for its greater control complexity.

\label{supp:experiment-cavity}
Circuit and cavity QED are well suited to both the one-excitation and finite-Fock protocols because they provide photon-number-selective rotations, photon-number-sensitive readout, controllable amplitude damping, and classical output randomization.  Photon-number-selective manipulations have been demonstrated in cavity systems \cite{Heeres2015}.  A possible experiment would use two bosonic modes, or two multimode symmetric registers, coupled to ancillary nonlinear elements.  The ancillas would prepare the encoded Bell state, implement the required selective rotations, and map the final Fock populations to binary readout outcomes.  A tunable attenuator, engineered reservoir, or calibrated waiting time would set the loss probability.  Repeating the experiment for the four pairs of settings would provide the correlators entering \(S\).

For the one-excitation protocol, the required control is a selective rotation between the vacuum and the one-photon state.  Its main challenge is to suppress unwanted coupling to higher Fock levels while maintaining a short measurement time.  Circuit and cavity systems offer strong spectral selectivity, but remote preparation of a high-fidelity encoded Bell state, phase stabilization between the two sites, and independent readout calibration remain necessary.  If the two parties are spatially separated for a Bell test, the measurement choices and readout windows must also satisfy the usual locality and detection requirements.  The present analysis concerns the state and measurement statistics after loss and does not by itself close those experimental loopholes.

A useful phenomenological description of the local transverse visibility is
\begin{equation}
v_{\rm exp}\simeq
(1-\epsilon_{\rm rot})
(1-\epsilon_{\rm meas})
(1-\epsilon_{\rm leak})
e^{-T/T_2}
V_{\rm state},
\label{eq:supp-visibility-budget}
\end{equation}
where \(\epsilon_{\rm rot}\) is the rotation error, \(\epsilon_{\rm meas}\) is the readout-assignment error, \(\epsilon_{\rm leak}\) is the probability of leaving the intended measurement subspace, \(T\) is the control and readout duration, \(T_2\) is the relevant coherence time, and \(V_{\rm state}\) is the initial state-preparation contrast.  This expression is bookkeeping rather than a microscopic noise model; the factors need not be independent.  It nevertheless identifies the quantities that should be calibrated experimentally.  For an off-resonant selective pulse with Rabi frequency \(\Omega\) and nearest unwanted detuning \(\Delta\), a common perturbative estimate is
\begin{equation}
\epsilon_{\rm leak}\sim(\Omega/\Delta)^2.
\label{eq:supp-leakage-estimate}
\end{equation}
Reducing \(\Omega\) suppresses leakage but lengthens the pulse and increases exposure to decoherence.  The useful operating point therefore balances spectral selectivity against the factor \(e^{-T/T_2}\).

The anchored finite-Fock protocol replaces the one-excitation code by
$\ket{0_L}=\ket0$,
 and 
$\ket{1_L}=\ket{\Nfp}$.
In a two-mode implementation, the source prepares
$\frac{\ket{00}+\ket{\Nfp\Nfp}}{\sqrt2}$,
after which each mode undergoes the calibrated pure-loss channel.  The longitudinal measurement distinguishes the vacuum from the nonvacuum levels.  The transverse observable
$X_{\Nfp}
=
\ket0\!\bra{\Nfp}
+
\ket{\Nfp}\!\bra0$
requires a phase-sensitive measurement of the coherence between \(\ket0\) and \(\ket{\Nfp}\).  One implementation is a selective rotation in this two-dimensional subspace followed by number-resolved detection.  Outcomes on the intermediate levels \(\ket1,\ldots,\ket{\Nfp-1}\), which lie in the kernel of \(X_{\Nfp}\), are assigned by unbiased classical randomization.  The resulting binary statistics directly realize the observable used in the analytical calculation, without first reconstructing the encoded qubit.

The anchored protocol has a lower asymptotic survival threshold than the one-excitation construction.  It is therefore preferable when loss tolerance is more important than excitation cost and when selective Fock-state control is available.  The improvement has a clear physical origin: increasing \(\Nfp\) changes the competition between the surviving two-party coherence \(\eta^{2\Nfp}\) and the population imbalance produced by complete erasure.  The same increase in \(\Nfp\), however, makes both state preparation and measurement more demanding.  The initial Bell state contains an \(\Nfp\)-excitation coherence on each side, and the local measurement must preserve and resolve the phase between Fock levels separated by \(\Nfp\).  Selective pulses may become slower, the nearest unwanted transitions may become harder to isolate, and preparation errors can accumulate with excitation number.  Moreover, the CHSH excess becomes small close to \(\etaG\), so resolving the asymptotic threshold requires correspondingly precise state preparation, loss calibration, and readout.

The visibility analysis in Sec.~\ref{supp:visibility} makes this trade-off quantitative.  If the coherent-pair visibility remains bounded below by a fixed \(v>0\) as \(\Nfp\) grows, imperfect measurements increase the finite excitation number required for violation but do not change the asymptotic boundary \(\etaG\).  If the visibility \(v_{\Nfp}\) decreases exponentially with \(\Nfp\), the measurement penalty competes directly with the surviving coherence and can shift the boundary.  An experiment should therefore measure \(v_{\Nfp}\) for each implemented excitation number rather than assume an \(\Nfp\)-independent visibility.  The anchored protocol is most attractive at moderate \(\Nfp\), where its loss advantage is already visible but the selective transition can still be calibrated accurately.

This construction is related to bosonic loss codes because it uses the same finite-Fock structure of the pure-loss channel \cite{Michael2016}.  Its operational goal is different.  Bosonic error correction normally extracts an error syndrome and applies a recovery intended to reconstruct the logical state.  Here the lossy output is measured directly.  The relevant observables are number measurements and selected Fock-level coherence probes, and their statistics are used immediately in the Bell functional.  This distinction is important experimentally: omitting recovery reduces the number of active operations, but it transfers the burden to the final measurement, which must resolve the logical coherence across the loss-populated Fock sectors.

\label{supp:experiment-sparse-pi}
The sparse binomial PI protocol takes a different route.  Instead of increasing the excitation separation within a single bosonic mode, it introduces redundancy across a symmetric \(n\)-particle register.  The logical states are superpositions of Dicke levels separated by \(g\), and the measurement is adapted to the residue and parity structure of those levels.  After a trial, particle-number readout determines the survivor number \(m\), while Dicke-weight readout determines the surviving excitation number \(q\).  Since the initial particle number \(n\) is fixed, \(m\) also gives the number \(r=n-m\) of deletions, without revealing which particles were lost.  The longitudinal observable \(Z_g\) assigns its binary value from the reconstructed coarse index \(L(q)\).  It can therefore be implemented by resolving the Dicke weight and applying the prescribed modular binning.

The transverse contraction \(\ContractX_g\) is more difficult.  It requires coherent measurements between pairs of surviving Dicke states whose weights differ by \(g\),
$\ket{D_q^m}
\longleftrightarrow
\ket{D_{q+g}^m}$.

A possible implementation would use a sequence of collective pulses, ancilla-mediated interactions, or digitally synthesized multilevel rotations to map the symmetric and antisymmetric combinations of each pair onto distinguishable populations.  Number readout would then complete the measurement.  Unpaired Dicke levels and branches with more than \(t\) deletions must produce unbiased random outcomes, as required by the theoretical construction.  The experimentally relevant quantities are therefore the probability of remaining in the good region \(r\leq t\), the assignment fidelity of \(Z_g\), and the coherence contrast of each \(q\leftrightarrow q+g\) pair.  These quantities determine the longitudinal and transverse correlators appearing in Theorem~\ref{thm:supp-sparse-stochastic}.

The main advantage of the sparse PI protocol is that it gives a rigorous finite-size construction for unflagged deletion and does not use the identities of the lost particles.  It makes the role of redundancy explicit and remains well defined on every survivor-number sector.  Its main limitation is experimental complexity.  Preparing the binomial superposition requires coherent control of several Dicke weights, and \(\ContractX_g\) couples levels separated by \(g\).  In the square family, \(g\sim\sqrt n\), so the transverse measurement becomes a high-order collective operation as the register grows.  The proven stochastic-loss scaling \(1-\eta=O(n^{-1/2})\) is also less favorable than the constant-loss thresholds obtained in the finite-Fock limit.  The sparse construction is therefore preferable when one wants a finite-\(n\) deletion protocol with a rigorous analytical guarantee and has access to collective symmetric-state control.  The one-excitation and anchored protocols are preferable when simpler low-dimensional measurements or constant-loss operation are the main priorities.

\label{supp:experiment-photonics}
Photonic platforms provide a natural setting for the one-excitation and anchored protocols because attenuation is directly implemented with beam splitters or lossy channels.  Photon-number-resolving detection gives the required population information, while displacement operations, phase references, and number-resolved detection can approximate single-rail coherence measurements \cite{Banaszek1999,Lee2017}.  For the one-excitation protocol, a path, time-bin, or cavity mode can encode the vacuum--one-photon subspace.  The main advantages are the natural loss model, rapid setting choice, and spatial separation of the parties.  The main limitations are probabilistic state preparation, detector inefficiency, phase instability, and the difficulty of realizing an exact projective \(X\) measurement with passive linear optics and vacuum ancillas alone.  Every no-click event must be assigned a valid outcome rather than removed by postselection.

The anchored protocol is more demanding in photonics because it requires coherent preparation of vacuum and \(\Nfp\)-photon components on both sides, together with a measurement of their relative phase.  Multiphoton generation rates, number resolution, and phase stability generally worsen with \(\Nfp\).  Photonics is therefore most competitive for small \(\Nfp\), where the pure-loss model is accurate and the required coherence remains experimentally accessible.  The sparse PI protocol could instead use symmetric multiphoton or multi-path Dicke states.  Such states have been prepared and manipulated experimentally \cite{Prevedel2009}, but scaling the state preparation and implementing coherent \(q\leftrightarrow q+g\) contractions would require nonlinear, measurement-induced, or adaptive resources beyond simple counting.

\label{supp:experiment-ions}
Trapped ions offer a complementary set of advantages.  Collective interactions and high-fidelity state-dependent fluorescence are well suited to preparing symmetric Dicke states and resolving excitation number \cite{Hume2009}.  The one-excitation logical Bell state can be prepared with global or sideband-mediated operations, and the longitudinal observable follows directly from fluorescence detection.  A selective transverse measurement can be implemented by mapping the vacuum--one-excitation coherence to populations before readout.  Compared with photonics, trapped ions provide stronger deterministic control and higher readout fidelity, but a physical deletion channel is less natural.  Loss may instead be simulated by removing selected ions, transferring them to states excluded from the detected register, or tracing out controlled subsystems.  Care is then required to ensure that the resulting operation matches unflagged deletion rather than flagged erasure.

The sparse PI protocol is conceptually well matched to trapped-ion collective control because both the code and the measurements remain in the symmetric Dicke basis.  Its modular \(Z_g\) measurement can be built from excitation-number readout and classical binning.  The transverse contraction remains the main obstacle: coupling \(\ket{D_q^m}\) selectively to \(\ket{D_{q+g}^m}\) may require high-order sideband processes or long digital pulse sequences.  Leakage to other Dicke levels, motional heating, and the growth of the pulse duration with \(g\) can reduce the effective coherence.  Trapped ions are therefore attractive for small proof-of-principle PI instances, where the branch structure and randomized boundary completion can be tested directly, but large square-family instances would require substantial improvements in multilevel collective control.

The three protocols expose a common physical trade-off.  The one-excitation code minimizes preparation and measurement complexity but has the highest survival threshold.  The anchored finite-Fock family lowers the asymptotic threshold by storing the relevant coherence across a larger excitation separation, at the cost of more demanding state preparation and selective readout.  The sparse PI construction replaces this excitation-space strategy with explicit particle redundancy and gives a finite-size guarantee for unflagged deletion, but requires coherent control across several Dicke sectors.  Across all platforms, the decisive experimental task is not number detection alone but the direct measurement of the coherence that survives the loss channel.  This is the operational content of the recovery-free approach developed in the main text: active reconstruction is omitted, while the accessible post-loss coherence is converted directly into the transverse correlations needed for CHSH violation.  No quantitative platform threshold is claimed here, but the required calibrations and correlators provide a concrete route for testing the theoretical hierarchy.

\end{widetext}

\end{document}